\documentclass[11pt]{article}

\usepackage{amsmath,amssymb,epsfig,cite,color,verbatim,array,bm,amsfonts,enumerate,enumitem,mathabx,amsthm}
\usepackage{graphics,float,multirow}
\usepackage[colorlinks=true,citecolor=ForestGreen,urlcolor=blue,linkcolor=blue]{hyperref}

\usepackage{blindtext}
\usepackage{microtype}
\usepackage[latin9]{inputenc}
\usepackage{authblk,booktabs,hhline}
\usepackage[table,dvipsnames]{xcolor}
\usepackage{graphicx}
\usepackage{tikz}
\usepackage{centernot}

\definecolor{light-gray}{gray}{0.9}

\theoremstyle{plain}
\newtheorem{theorem}{Theorem}
\newtheorem*{theorem*}{Theorem}
\newtheorem{definition}{Definition}
\newtheorem{proposition}{Proposition}
\newtheorem{definitionTwice}{Definition}[section]

\newtheorem{lem}{Lemma}
\numberwithin{equation}{section}
\allowdisplaybreaks
\DeclareMathOperator{\Hessian}{Hess}

\title{Foundations of
ghost stability}

\author[1,2,3]{Ver\'onica
Errasti D\'iez
\footnote{veronica.errasti@cunef.edu}}

\author[3]{Jordi Gaset Rif\`a
\footnote{jordi.gaset@cunef.edu}}

\author[4,5]{Georgina
Staudt
\footnote{georgina@mpp.mpg.de}}

\affil[1]{ {\small {\it
Universit\"ats-Sternwarte,
Fakult\"at f\"ur Physik,
Ludwig-Maximilians-Universit\"at M\"unchen, Scheinerstra\ss e 1,
81679 M\"unchen,
Germany}}}

\affil[2]{ {\small {\it
Excellence Cluster ORIGINS,
Boltzmannstra\ss e 2,
85748 Garching,
Germany}}}

\affil[3]{ {\small {\it
Department of Mathematics,
CUNEF Universidad, 
Calle Pirineos 55,
28040 Madrid,
Spain}}}

\affil[4]{ {\small {\it
Max-Planck-Institut f\"ur Physik
(Werner-Heisenberg-Institut),
Boltzmannstra\ss e 8,
85748 Garching,
Germany}}}

\affil[5]{ {\small {\it 
Arnold-Sommerfeld Center for Theoretical Physics,
Theresienstra\ss e 37,
80333 M\"unchen,
Germany}}}

\date{}

\begin{document}

\maketitle

\begin{abstract}
We present a new method
to analytically prove global stability in ghost-ridden dynamical systems.
Our proposal encompasses all prior results
and consequentially extends them.
In particular,
we show that stability can follow from
a conserved quantity that is unbounded from below,
contrary to expectation.
Novel examples illustrate all our results.
Our findings take root on a careful examination
of the literature,
here comprehensively reviewed for the first time.
This work lays the mathematical basis
for ulterior extensions to field theory and quantization,
and it constitutes a gateway for inter-disciplinary research
in dynamics and integrability.
\end{abstract}

\section{Motivation}
\label{sec:intro}

Higher-order theories enjoy a long, fruitful history.
Archetypal examples include
the Pais-Uhlenbeck oscillator~\cite{PU50},
Podolsky electrodynamics~\cite{Podolsky:1942zz,Podolsky:1944zz,Podolsky:1945chv}
and the gravitational models
by Weyl~\cite{Weyl:1919fi,Weyl21}
and Lovelock~\cite{Lovelock:1971yv}.
Interest in higher-order settings is heterogeneous,
to say the least.
For illustration purposes,
we mention
geometry~\cite{Cartan,Gracia:1989ig,EtayoGordejuela:2007pt,Bruce:2014lxa,CastrillonGaset},
ultraviolet regularization~\cite{Stelle:1976gc,Grinstein:2007mp,Tomboulis:2015esa,Mannheim:2023mfp,Mannheim:2024ftu},
cosmology~\cite{Caldwell:1999ew,Nojiri:2003ft,Charmousis:2008kc,Brandenberger:2016vhg,DiValentino:2021izs}
and string theory~\cite{Freeman:1986zh,Polyakov:1986cs,McOrist:2012yc,Liu:2013dna,Calderon-Infante:2023uhz}.
Recent constructions
of novel higher-order settings~\cite{Kubo:2024ysu,Fring:2024brg,Kuzenko:2024zra,Colleaux:2024ndy,Mironov:2024ffx,Mandal:2024nhv,Mironov:2024umy}
evince topicality and augur continued zeal.

The so-called {\it Ostrogradsky theorem}~\cite{Ostrogradsky:1850fid,singOstr}
states that
higher-order Lagrangians
generically correspond to
Hamiltonians that are not bounded from below.
Such unbounded scenarios lie beyond the scope
of the Lagrange-Dirichlet theorem~\cite{Lag,Dir},
wherein strict minima of the energy
correspond to stable points of the dynamics,
in the Lyapunov sense~\cite{Lya}
--- definition \ref{def:L2} later on.
For a field theoretical treatment, see~\cite{Biro:1994bi}.

In pursuit of stability within the higher-order arena,
the most heavily transited way around the hurdle
consists in constraining 
the equations of motion
so as to restore a lower bound on the Hamiltonian.
Hefty efforts have been made
in this regard~\cite{Chen:2012au,Diaz:2014yua,Klein:2016aiq,Crisostomi:2017aim,Diaz:2017tmy,Ganz:2020skf,ErrastiDiez:2023gme}
with remarkable outcomes,
particularly in Modified Gravity. 
Influential representatives encompass
$f(R)$ theories~\cite{Sotiriou:2008rp,DeFelice:2010aj}
and de~Rham-Gabadadze-Tolley (dRGT)
massive gravity~\cite{deRham:2010kj}.
Recall that caution must be exercised
in such serpentine scenarios~\cite{Gorbunov:2005zk,BeltranJimenez:2019acz},
especially in the absence of an intrinsic definition
of the claimed pathology~\cite{Wolf:2019hzy}.
We refer to~\cite{Delhom:2022vae}
for a pedagogical presentation.

A prevalent misconception hampers
alternative progress.
The fallacious {\it Ostrogradsky instability}
affirms that the absence of a lower bound for the Hamiltonian
entails the instability of the dynamical system.
This is in direct contradiction with
a substantial amount of counterexamples,
e.g.~\cite{Pagani:1987ue,Smilga:2013vba,Kaparulin:2015owa,Abakumova:2018eck,Damour:2021fva,Deffayet:2023wdg,Heredia:2024wbu}.
For their clarity and vehemence
against the faulty identification
of Ostrogradsky theorem with instability,
we advert to~\cite{Heredia:2024wbu,Smilga:2017arl,Pavsic:2020aqi}.
In this work,
we shall exploit the breach and investigate
the classical stability of systems
whose Hamiltonian
is manifestly not bounded from below.
For concision,
we shall employ
the widespread nickname {\it ghost systems}.

Before proceeding ahead and for completeness, 
we register the analogy 
between the higher-order and nonlocal frameworks.
To a lesser but far from negligible extent,
nonlocal theories
share the features of
being well-established~\cite{Barnaby:2007ve,Heredia:2021wja,Heredia:2022mls},
possessing multi-disciplinary scope~\cite{Deser:2007jk,Biswas:2010zk,Calcagni:2013eua,Calcagni:2023goc}
and deserving persistent attention~\cite{Deffayet:2024ciu}.
A prominent exemplar is
the Wess-Zumino-Witten model~\cite{Wess:1971yu,Novikov:1982ei,Witten:1983ar}.
Nonlocal theories tend to be associated with
a Hamiltonian that is not bounded from below,
which again does not imply
their dynamical instability.
The scrutiny of nonlocal theories
calls for tailored techniques
beyond their loose understanding
as infinite-order settings,
see~\cite{Gomis:2000sp,Heredia:2023cgs}
and references therein.
Hence,
we shall not expressly consider them.
Notwithstanding,
the essence of our account 
and conclusions apply therein.

\paragraph{Organization.}
The paper is arranged as follows.
In section \ref{sec:review},
we classify the notions of classical stability
that have been put forward
for ghost systems.
Additionally,
we provide a succinct yet all-encompassing
overview of the literature.
In section \ref{sec:const},
we formalize
the global subset of the said classical stability notions.
With a focus on one such formalization,
in section \ref{sec:Qs} we expose
the mathematical foundation of known stable ghost systems,
and develop a framework for novel settings,
all of which are exemplified in the subsequent section \ref{sec:ex}.
We summarize and draw our conclusions in the final section \ref{sec:concl}.
Tangential technicalities are consigned to
appendices \ref{app:thenergy}-\ref{app:properties}.

\section{State of the art: multifacetedness of ghost stability}
\label{sec:review}

Granted that ghost systems can be stable, 
the question arises as to what is meant by their stability.
Indeed,
beyond the misreckoning
of unbounded Hamiltonians
as necessarily yielding unstable dynamics,
we identify a secondary hindrance
in the multiplicity
of stability notions currently in use
for such systems.
Accordingly,
we proceed to their categorization.
We restrict attention to the classical realm.

\paragraph{Global notions of stability.}
These apply to all solutions, irrespective of initial conditions.
\begin{enumerate}[label=\textcolor{blue}{G\arabic*.}, ref=G\arabic*]
\item
\label{def:G1}
Lagrange stability~\cite{Lag,LaSalle}.
Solutions are bounded.
\item
\label{def:G2}
Benign ghost~\cite{Smilga:2013vba,Robert:2006nj}. 
Solutions are either bounded
or diverge only as time tends to infinity (runaway).
This excludes solutions
that diverge within finite time (blow ups).
\end{enumerate}

\paragraph{Local notions of stability.}
These apply to some solutions,
characterized by initial conditions.
\begin{enumerate}[label=\textcolor{blue}{L\arabic*.}, ref=L\arabic*]
\item
\label{def:L1}
Island of stability~\cite{Smilga:2004cy,Smilga:2005gb}.
Local counterpart to either of the above global notions.
\item
\label{def:L2}
Lyapunov stability~\cite{Lya} and generalizations thereof. 
Notions of stability in the immediate vicinity of
a critical point of the system.
\end{enumerate}

By definition,
\ref{def:G1} $\subset$ \ref{def:G2}.
We remark that,
given a dynamical system
with free parameters and/or functions,
any fixed notion of stability
is typically attained only
within a certain parametric and/or functional region.

\subsection{Main classical results}
\label{sec:reviewclass}

For intelligibility and contextualization,
we here provide
a necessarily abridged recount of the cardinal outcomes
within each stability notion.
As a preliminary observation,
we note that
all notions
originate and are employed 
beyond our targeted ghost systems.

\ref{def:G1} studies flourished
motivated by engineering applications
wherein Lyapunov stability \ref{def:L2}
is overrestrictive~\cite{Gyfto}. 
In particular,
\ref{def:G1} stability is a befitting property
of dynamical systems intended to model
situations where variables need not be kept within a stringent tolerance
for an acceptable performance.
As regards ghost systems,
\ref{def:G1} is a prolific
entry in our stability dictionary.
Kaparulin and coauthors have developed
a vast bodywork,
with onset~\cite{Kaparulin:2014vpa}
and culmination~\cite{Abakumova:2018eck},
encompassing~\cite{Kaparulin:2015owa,Kaparulin:2015uxa,Kaparulin:2015pda,Abakumova:2017syd,Kaparulin:2017swa,Abakumova:2017uto,Kaparulin:2017uar,Kaparulin:2018npv,Abakumova:2019ifi,Kaparulin:2019njc,Abakumova:2019wpn,Abakumova:2019dov,Kaparulin:2020gqn}.
These works prove \ref{def:G1} stability
in numerous mechanical and field theoretical settings
of patent physical relevance and complexity.
Central to
the clearly prescribed yet not fully systematized
construction of such theories
is their bi-Hamiltonian character.
(Bi-Hamiltonian systems were pioneered
by Magri~\cite{Magri:1977gn}.
A cogent introduction can be found in chapter 7.3 of ~\cite{Olver}.)
In broad strokes, 
the expounded frameworks possess
a ghost-like Hamiltonian,
but are also associated with a non-canonical Hamiltonian
that is bounded from below.
The latter warrants \ref{def:G1} stability.
The method requires
such a compendium of attributes
from a root system
and proves its preservation 
in non-trivially extended scenarios.
For an articulate summary,
we recommend the introduction in~\cite{Kaparulin:2020rqz}.
On the other hand,
Deffayet et al.~\cite{Deffayet:2023wdg,Deffayet:2021nnt} have shown \ref{def:G1} stability
for a broad class of mechanical integrable systems.
Along these lines,
Heredia and Llosa
have incorporated further
examples in their second version of~\cite{Heredia:2024wbu}.

In certain settings,
\ref{def:G2} stability can be assimilated to the completeness
of the Hamiltonian vector field
--- see the discussion below definition \ref{def:G2formal}
in the subsequent section \ref{sec:const}.
In this sense,
\ref{def:G2} is a classic subject of study in differential geometry, e.g.~\cite{HopfRinow,Palais:1957}.
Applied to ghost systems,
\ref{def:G2} stability has been
the goal of significant sundry efforts
by Smilga and collaborators.
Prefatory studies~\cite{Smilga:2004cy,Smilga:2005gb,Ivanov:2005qf,Smilga:2005pr,Smilga:2006tu,Smilga:2006ax}
ultimately led to \ref{def:G2} stable supersymmetric models:
mechanical at the outset~\cite{Robert:2006nj}
and field theoretical later~\cite{Smilga:2013vba}.
A lucid chronicle is~\cite{Smilga:2017arl},
while
forefront results
are conveyed in~\cite{Damour:2021fva}.
Therein,
\ref{def:G2} stable systems
are systematically obtained
from non-ghost \ref{def:G1} stable systems,
without appeal to supersymmetry.
The construction 
incorporates~\cite{Robert:2006nj,Smilga:2013vba}
as subcases.

(The kind of) Global stability in some systems
awaits future advancements.
In pursuit of \ref{def:G2} stability,
Smilga suggested
a proliferation of examples~\cite{Smilga:2020elp},
including the renowned Toda chain~\cite{Toda}
and the Korteweg-de Vries (KdV) system~\cite{Korteweg:1895lrm},
with time and space interchanged.
Both Toda and especially KdV
are epitomic instances of bi-Hamiltonian systems.
In spite of ulterior progress~\cite{Fring:2024brg,Damour:2021fva,Smilga:2021nnx}
in relation to time-space flipped
KdV and certain modifications thereof,
\ref{def:G1} and/or \ref{def:G2} stability 
remain inconclusive.

As local correlative to \ref{def:G1},
ghostly \ref{def:L1} stability counts two apogees.
First,
in an ambitious generalization
of their frame,
Kaparulin-Lyakhovic-Nosyrev~\cite{Kaparulin:2020rqz}
undertake the stabilization of
fully unbounded from below bi-Hamiltonian root systems
via interacting augmentations.
As a result,
they prove \ref{def:L1} stability
starting with two such root systems:
the resonant Pais-Uhlenbeck
oscillator
and zero mass Podolsky electrodynamics.
This scheme does not generalize to \ref{def:G1} stability,
at least not straightforwardly.
Second,
following the thread of~\cite{Damour:2021fva,Smilga:2020elp,Smilga:2021nnx},
Fring-Taira-Turner~\cite{Fring:2024brg}
have recently proven
the oscillatory behavior of
a number of modified and time-space swapped KdV systems,
for specific initial conditions.
We already indicated that \ref{def:G1} stability
persists an enticing open question.
Complementarily,
as local incarnation of \ref{def:G2} stability,
\ref{def:L1} stability
has made an appearance
in a variety of cosmological contexts
benefitting from ghost states,
e.g.~\cite{Linde:1988ws,Carroll:2003st,Kaplan:2005rr,Salvio:2019ewf,Gross:2020tph}.
We single out~\cite{Smilga:2004cy,Smilga:2005gb}
as decisive for our present discussion,
since they triggered
the very notion of \ref{def:G2} stability.

Lyapunov stability \ref{def:L2} and refinements thereof,
particularly its asymptotic variant,
are vertebral to dynamical systems.
Within scope,
we have been able to trace back \ref{def:L2}
studies in ghost systems to~\cite{Pagani:1987ue}.
The latest contribution~\cite{Heredia:2024wbu}
includes proofs of \ref{def:L2} stability for
nonlocal mechanical systems.

Numerical studies~\cite{Pavsic:2012pw,Pavsic:2013noa,Pavsic:2013mja,Pavsic:2016ykq,Fring:2023pww}
have been instrumental
to the development of the above reviewed
analytic results
and continue to motivate enticing
questions~\cite{Fring:2024xhd}.

\subsection{Comments on quantization}
\label{sec:reviewquant}

Although beyond scope,
for completeness
we here provide a lightning recap
on quantum ghost settings.

Sustained consensus 
identifies
unitarity as the paramount feature
for quantum viability
and the existence of a well-defined ground state
as the embodiment of quantum stability.
A priori in clash,
the canonical quantization of generic ghost settings
involves a dichotomy:
negative norm states or
a spectrum that is not bounded from below~\cite{Fring:2024brg,Smilga:2013vba,Smilga:2017arl,Robert:2006nj,Barth:1983hb,Weldon:2003by,Smilga:2008pr,Ilhan:2013xe,Woodard:2015zca,Salvio:2015gsi}.
Early works,
such as Stelle's influential proof of renormalizability
for fourth-order quantum gravity~\cite{Stelle:1976gc}
or its reconsideration in~\cite{Barth:1983hb},
tend to focus on the reconciliation
of unitarity and negative norm states.
Recently, 
the study of negative eigenvalues
has gathered attention,
as emphatically defended in~\cite{Woodard:2015zca}.

In this regard,
Smilga~\cite{Smilga:2017arl,Smilga:2008pr},
also with Robert~\cite{Robert:2006nj},
surmounts the paradox
in certain cases,
proving unitarity of the evolution operator
while minutely characterizing
the unorthodox, unbounded from below spectrum.
Their analyses heavily rely on
solvability,
direct or by means of
Liouville integrability
via a non-trivial change of coordinates to action-angle variables.
With a solid basis on
the Pais-Uhlenbeck oscillator~\cite{Smilga:2008pr}
and their propounded
\ref{def:G2} stable mechanical system~\cite{Robert:2006nj},
a remarkable conjecture is put forward:
global stability of a ghost system
ensures its quantum viability.
More precisely,
both \ref{def:G1} and \ref{def:G2} stability
warrant quantum unitarity,
with pure point and band-wise continuous dense spectra,
respectively.

The exception to the rule are~\cite{Salvio:2015gsi}
and the very recent~\cite{Holdom:2024onr},
where negative norm states are shown to be
compatible with unitarity.
In~\cite{Salvio:2015gsi},
the approach is mechanical in nature.
It calls for
coordinate operators that are
odd under time reflection 
and also for a self-adjoint yet not Hermitian Hamiltonian.
See also~\cite{Strumia:2017dvt}.
In~\cite{Holdom:2024onr},
a new inner product is proposed
which, used in the Born rule,
provides a sensible probability interpretation
--- in particular, respecting unitarity.
This is worked out in detail in mechanical scenarios
and toy-models that support its validity in four dimensional field theory.
The focus is on classically unstable systems,
which aligns well with Smilga's conjecture.
In this regard, 
see~\cite{Romatschke:2024mxr} and references therein.
Complementarily,
we cannot omit the even more recent~\cite{Gies:2024dly},
wherein specific higher-order quantum field theories
are found to be stable at first loop order,
with the presumable tachyonic mass pole absent.
For greater insights into an allied setting,
see~\cite{Gies:2023cnd}.

Notwithstanding,
most efforts center around
the outright cancellation of the dichotomy.
We highlight three such possibilities.
By construction,
the bi-Hamiltonian theories
considered by Kaparulin and friends
admit canonical quantization
of a non-canonical Hamiltonian bounded from below,
with respect to appropriate
Poisson brackets.
In this manner,
positive norm states and
a ground state
are obtained.
The technique was
developed in~\cite{Kazinski:2005eb,Lyakhovich:2005mk,Lyakhovich:2006sc}.

Besides,
it has been suggested that
the Pais-Uhlenbeck oscillator be viewed as
invariant under Parity and Time reflection (PT) 
and the apposite non-canonical quantization scheme applied~\cite{Bender:2007wu,Bender:2008gh,Bender:2008vh}.
(An instructive recapitulation
of PT quantum mechanics
appeared not long ago~\cite{Bender:2023cem}.)
Yet again,
the result is
positive norm states and a ground state.
On the not so bright side,
PT symmetry is obtained through
a rather unnatural analytic continuation
and the proposal does not straightforwardly extend to
the non-trivial interacting scenario~\cite{Smilga:2008pr}.
To the best of our knowledge, 
such criticisms have never been addressed.

Lastly,
a popular circumvention
consists in the development of
ghost confinement mechanisms, e.g.~\cite{Kawasaki:1981gk,Arkani-Hamed:2003pdi,Mukohyama:2009rk,Donoghue:2017fvm,Frasca:2022gdz,Liu:2022gun}.
These techniques remove
the modes responsible for
the bottom unboundedness of the Hamiltonian
from the physical spectrum,
including the classical regime.

\section{Definitions: the big picture of global stability}
\label{sec:const}

Additional progress in the above reviewed subject
of stable ghost systems necessitates,
or at least greatly benefits from,
a rigorous comprehension of the various stability notions at play.
Therefore,
this section is devoted to the formalization
of the comparatively under-examined
global stability notions \ref{def:G1} and \ref{def:G2}.
Local stability \ref{def:L1}
is then recovered for particular sets
of solutions, characterized by initial conditions.
Our work thus complements
the long-standing formalization
of Lyapunov stability \ref{def:L2}
and plentiful generalizations thereof.
We underscore the generality of the forthcoming discussion,
which shall be explicitly applied to (certain) ghost systems
in section \ref{sec:ex}.
The domain and interrelation
of our proposed definitions
are succinctly conveyed in figure~\ref{fig:G1s}.

We set our scope to dynamical systems.
We require their generic characterization,
independent of a Lagrangian
and/or Hamiltonian.
The rationale is three-fold.
First, 
this entails the same level of abstraction
employed by Lyapunov.
Second,
numerous dynamical systems exist
for which a standard Lagrangian
and/or Hamiltonian cannot be defined.
Well-known examples include
the Lorenz system~\cite{Lorenz:1963yb}
and nonlocal actions~\cite{Heredia:2024wbu}.
Third,
we thus remove potential ambiguities
that may arise 
in bi-Hamiltonian systems~\cite{Smilga:2008pr},
such as the extensively researched
Pais-Uhlenbeck oscillator~\cite{PU50}.
Accordingly,
we consider
the set of solutions:
\begin{align}
\label{eq:soldef}
\textrm{Sol}=\{\gamma: I_\gamma\subset\mathbb{R}\rightarrow \mathcal{P} \;\textrm{ continuous} \},
\end{align}
for some topological space $\mathcal{P}$;
henceforth taken to be at least Hausdorff.
Here,
$I_\gamma$ is a fixed (non-empty) interval of time
where the curve $\gamma$ is defined.
For example,
given a Lagrangian system, 
$\mathcal{P}$ is the tangent space of the configuration space, 
$\textrm{Sol}$ are the solutions of the Euler-Lagrange equations and
$I_\gamma$ is their maximal interval of definition.

In this broad setting,
the notion \ref{def:G1}
admits two formalizations,
depending on $\mathcal{P}$.
The definition
closest to the understanding employed in the ghost literature is:

\setcounter{section}{1}

\begin{definitionTwice}
\label{def:G1bounded}
For $\mathcal{P}$ metric,
a system is \textbf{G1 in the bounded sense} if,
for all curves $\gamma\in \emph{Sol}$,
their image $\gamma(I_\gamma)$
is bounded.
\end{definitionTwice}

\noindent
Notice that the metric nature of $\mathcal{P}$ 
allows for a notion of distance,
which in turn renders boundedness meaningful.
We here propose 
the broader notion:

\begin{definitionTwice}
\label{def:G1compact}
For $\mathcal{P}$ Hausdorff,
a system is \textbf{G1 in the compact sense} if,
for all curves $\gamma\in \emph{Sol}$,
their image $\gamma(I_\gamma)$
is relatively compact with respect to $\mathcal{P}$.
\end{definitionTwice}

\setcounter{section}{3}

Recall that
a subset $U$ of a topological space $\mathcal{P}$ is
relatively compact with respect to $\mathcal{P}$ if
its closure $\overline{U}$ is compact.
Moreover,
for $\mathcal{P}$ Hausdorff,
$U$ is relatively compact with respect to $\mathcal{P}$ iff it is inside a compact set.
(The definition of Hausdorff, relatively compact and allied concepts and properties
can be found in e.g.~\cite{Schechter},
see 17.15.)

Two important observations follow.
First, all metric spaces are Hausdorff,
but the converse is not true,
see e.g.~11.5 in~\cite{SteenSeebach}.
This implies that the scope of definition \ref{def:G1compact}
is larger than that of definition \ref{def:G1bounded}.
Second, for $\mathcal{P}$ metric (and hence Hausdorff),
definitions \ref{def:G1bounded} and \ref{def:G1compact}
are not equivalent.
In this case,
definition \ref{def:G1compact} implies definition \ref{def:G1bounded},
but the converse is not true.
Equivalence is achieved by requiring
the Heine-Borel property for $\mathcal{P}$~\cite{WJ:1987},
which states that
any closed and bounded subset of $\mathcal{P}$ is compact. 
For a counterexample,
the renowned Riesz's lemma shows that the closed unit ball of infinite-dimensional
normed vector spaces is not compact
--- see e.g.~5.2 in~\cite{Lax}.

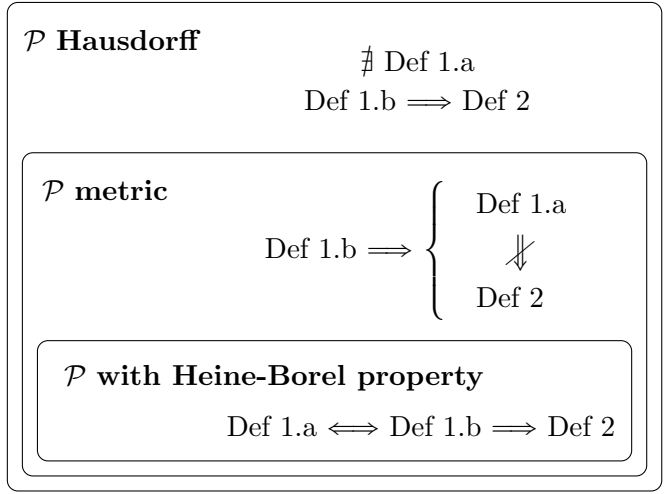
\begin{figure}
\begin{center}
\begin{tikzpicture}
\draw[rounded corners] (-6, 0) rectangle (2.7, 6.5) {};
\node at (-4.6,6){$\mathbf{\mathcal{P}}$ \textbf{Hausdorff}};
\node at (-0.55,5.2){Def $1$.b $\Longrightarrow$ Def $2$};
\node at (-0.55,5.7){$\nexists$ Def $1$.a};
\draw[rounded corners] (-5.8, 0.2) rectangle (2.5, 4.5) {};
\node at (-4.7,4){$\mathbf{\mathcal{P}}$ \textbf{metric}};
\node at (-0.55,3.2){Def $1$.b $\Longrightarrow \begin{cases}
&\text{Def }1\text{.a}
\\
&\quad\rotatebox[origin=c]{-30}{\big\slash}\hspace*{-0.35cm}{\big\Downarrow}
\\
&\text{Def }2
\end{cases}$
};
\draw[rounded corners] (-5.6, 0.4) rectangle (2.3, 2) {};
\node at (-2.45,1.5){$\mathbf{\mathcal{P}}$ \textbf{with Heine-Borel property}};
\node at (-0.5,0.87){Def $1$.a $\Longleftrightarrow$ Def $1$.b $\Longrightarrow$ Def $2$};
\end{tikzpicture}
\caption{Visualization of
the scope of 
global stability
definitions \ref{def:G1bounded}, \ref{def:G1compact}
and \ref{def:G2formal},
together with relations between them.
These formalize the preexisting notions \ref{def:G1} and \ref{def:G2}.
}
\label{fig:G1s}
\end{center}
\end{figure}

Beyond its greater scope,
section \ref{sec:Qs}
(particularly our main result theorem \ref{th:conf})
unequivocally leads to our prevailing 
of definition \ref{def:G1compact} over definition \ref{def:G1bounded}.
More generally,
it is essential to keep watch on
metric $\mathcal{P}$'s without the Heine-Borel property,
since they are physically relevant.
Conspicuous instances include 
infinite-dimensional Banach and Hilbert spaces
(as particular exemplars of infinite-dimensional normed vector spaces)
and geodesically incomplete Riemannian manifolds;
for example, $\mathbb{R}^n\setminus\{0\}$ with the Euclidean metric.
In all such scenarios,
it becomes indispensable
to pick out between definitions
\ref{def:G1bounded} and \ref{def:G1compact}.

Moving on to
the stability notion \ref{def:G2},
we observe that this is a weaker version of \ref{def:G1}.
Specifically,
\ref{def:G2} accepts solutions that diverge at infinite time.
For any such $\gamma \in \textrm{Sol}$,
the image $\gamma(I_\gamma)$
is not inside a compact set.
Nevertheless,
\ref{def:G2} excludes solutions
that diverge at a finite time.
In order to detect these $\gamma$'s
and banish them from the formal definition of \ref{def:G2},
we follow the work of Palais~\cite{Palais:1957}
and probe spans of time
that are relatively compact with respect to $\mathbb{R}$,
instead of probing the whole time range.
Accordingly,
we put forward:

\setcounter{definition}{1}

\begin{definition}
\label{def:G2formal}
For $\mathcal{P}$ Hausdorff,
a system is \textbf{G2} if,
for all curves $\gamma\in \emph{Sol}$,
and for all subsets $S\subset I_\gamma$
relatively compact with respect to $\mathbb{R}$,
the image $\gamma(S)$ is relatively compact
with respect to $\mathcal{P}$.
\end{definition}

If Sol is the set of maximal integral curves of a smooth\footnote{Classical texts
consider $X$ smooth.
For the purists,
we note it is enough to require that $X$ is ${C}^1$.} vector field $X$,
then by negation of the Escape Lemma (see, e.g.~\cite{Lee}),
the system is G2 iff $X$ is complete.
This is the case of regular Lagrangian and/or Hamiltonian systems.
Bear in mind that a vector field is complete
if its integral curves can be defined for all times:
$I_\gamma=\mathbb{R}$ in (\ref{eq:soldef}). 

It is obvious that definition \ref{def:G1compact}
implies definition \ref{def:G2formal},
but the converse is not true.
Moreover,
for $\mathcal{P}$ metric,
definition \ref{def:G1bounded}
does not imply \ref{def:G2formal}.

All our proposed definitions
\ref{def:G1bounded}, \ref{def:G1compact} and \ref{def:G2formal}
are manifestly intrinsic: 
they do not appeal to coordinates.

\section{Proving G1: the pivotal role of conserved quantities}
\label{sec:Qs}

Henceforth,
we concentrate on global stability G1,
in the compact sense of definition \ref{def:G1compact}.
In contraposition to G2 global stability
in the formal sense of definition \ref{def:G2formal},
our prevailed conceptualization a priori prevents the divergence
of wave functions and probability distributions
upon quantization,
which in principle span the whole configuration space.
This argument was also presented in~\cite{Ilhan:2013xe,Deffayet:2023wdg}
and is reinforced by
calculations~\cite{Robert:2006nj,Smilga:2017arl,Smilga:2008pr}
that signal a comparatively less unconventional
quantum structure.
In contraposition to G1 global stability
in the bounded sense of definition \ref{def:G1bounded},
our prevailed conceptualization
enjoys a broader scope in its definition
(all Hausdorff spaces)
and always implies G2 global stability.

In this section,
we investigate
a suitable tool kit for the implementation 
of definition \ref{def:G1compact}
in systems of the form (\ref{eq:soldef}).
We begin by briefly expounding the mathematical results
that support such G1 stability proofs
and then proceed to their unprecedented generalization.
To this aim,
the following comprehension of 
a conserved quantity
is employed hereafter.

\begin{definition}
\label{def:consquant}
Let $N$ be a topological space.
A continuous function $Q:\mathcal{P}\rightarrow N$
is a \textbf{conserved quantity} of {\emph{Sol}} if,
for any fixed curve $\gamma\in\emph{Sol}$,
the value of $Q$ is unchanged; that is, $Q(\gamma(t_1))=Q(\gamma(t_2))$ for all points $t_1,t_2\in I_\gamma$.
\end{definition}

Observe that our reference definition \ref{def:G1compact}
is topological in nature.
Befittingly,
definition \ref{def:consquant}
brings the concept of conserved quantity
to the same level of abstraction.

\subsection{When does energy work?}
\label{sec:energy}

A well-known resolution to the interrogation is as follows.

\begin{theorem}
\label{th:energy}
Given a system $\emph{Sol}$
with $\mathcal{P}=\mathbb{R}^n$,
if the energy $E:\mathbb{R}^n\rightarrow \mathbb{R}$
is a
twice differentiable
conserved quantity
that
has a minimum
and whose Hessian is positive definite everywhere,
then the system is \emph{G1}.
\end{theorem}

Since $\mathbb{R}^n$ is metric and has the Heine-Borel property,
the system is G1 in the sense of
both equivalent definitions \ref{def:G1bounded} and \ref{def:G1compact}.
The converse of theorem \ref{th:energy} is not true,
as abundantly exemplified in section~\ref{sec:review}.

For completeness,
a recount of the proof
is provided in appendix \ref{app:thenergy}.
(Alternatively,
theorem \ref{th:energy} can be viewed
as a corollary of theorem \ref{th:coercive},
also proved in appendix \ref{app:thenergy}.)
It is beneficial to bear in mind that the proof relies on two facts.
First, solutions $\gamma\in\textrm{Sol}$
are inside level sets of the energy $E$
because the latter is a conserved quantity.
Second,
the level sets of $E$ are bounded,
as a direct consequence of $E$ having both a minimum
and a positive definite Hessian everywhere
--- the very definition of which necessitates the
auxiliary property that $E$ is twice differentiable. 
To avoid confusion stemming from polysemy,
we note that we abide by the understanding of
level set of a function $f$
as the set of points in its domain
for which $f$ takes a fixed constant value.
Expressly,
for $f:M\rightarrow N$ and $c\in N$,
the level set is 
\begin{align}
\label{eq:levelset}
f^{\neg}(c):=\{x\in M \, | \, f(x)=c\}.
\end{align}

Notice that theorem \ref{th:energy}
does not ensure G1 stability of the system $\textrm{Sol}$
if either of the two requirements
on the conserved, twice differentiable energy $E$
is not met.
Namely, 
if either $E$ has a minimum but its Hessian is not positive definite somewhere
or
if $E$ does not attain a minimum
in spite of having a positive definite Hessian everywhere,
then the axioms of theorem \ref{th:energy} do not apply.
Respective counterexamples are the Hamiltonian systems on $\mathbb{R}^2$
\begin{align}
\label{eq:Hcounterth1}
H=\frac{1}{2}(p-x)^2,
\qquad 
H=\textrm{exp}(x^2+p),
\end{align}
which can be easily verified to each admit
a solution that is not relatively compact with respect to $\mathbb{R}^2$:
\begin{align}
\label{eq:solcounterth1}
\begin{cases}
x=t,\\
p=t+1,
\end{cases}
\qquad
\begin{cases}
x=t,\\
p=-t^2.
\end{cases}
\end{align}
We remark that (\ref{eq:Hcounterth1})
illustrate the fact that (in)stability
is a feature of generic dynamical systems,
including but transcending our targeted ghost systems.
(As quantum correlative,
it is interesting to mention
that unitarity loss can occur
in ghost-free settings~\cite{Asorey:2018wot}.)

\subsection{The case of ghost systems}
\label{sec:coercive}

An enticing overspill of the hypotheses of theorem \ref{th:energy}
is provided by (mechanical) ghost systems,
whose imprint is an energy that is not bounded from below.
Close inspection of the apposite proof in appendix \ref{app:thenergy}
readily reveals that any conserved quantity $Q$,
not necessarily the energy $E$,
can be employed in the theorem.
This cognizance underlies
the great majority of G1 stable ghost systems
known to date,
developed
by Kaparulin et al.~\cite{Kaparulin:2015owa,Abakumova:2018eck,Kaparulin:2014vpa,Kaparulin:2015uxa,Kaparulin:2015pda,Abakumova:2017syd,Kaparulin:2017swa,Abakumova:2017uto,Kaparulin:2017uar,Kaparulin:2018npv,Abakumova:2019ifi,Kaparulin:2019njc,Abakumova:2019wpn,Abakumova:2019dov,Kaparulin:2020gqn,Kaparulin:2020rqz}\footnote{Specifically and as already heeded,
these works avail themselves of a $Q$
that is another, inequivalent Hamiltonian of the system.}.

The remaining results,
those of Deffayet and friends~\cite{Deffayet:2023wdg,Deffayet:2021nnt}
and more recently by Heredia and Llosa~\cite{Heredia:2024wbu},
consider a more general kind of conserved quantity:
coercive instead of positive definite.
For self-containment,
we recall:

\begin{definition}
\label{def:limcoerc}
A function $f:\mathbb{R}^n\rightarrow \mathbb{R}$
is \textbf{coercive}
if the absolute value of $f(x)$ diverges as the modulus of $x$ diverges:
\begin{align}
\label{eq:limcoerc}
\lim_{\|x\|\rightarrow +\infty}|f(x)|\rightarrow +\infty\,.
\end{align}
\end{definition}

(En passant,
we point out that coercive functions
can in fact be defined for arbitrary metric spaces,
see lemma \ref{lem:bounded} in appendix~\ref{app:thenergy}.
Upon demanding
the Heine-Borel property,
coercive functions are identified with proper functions.
The double implication is established
in appendix \ref{app:proper},
and its relevance is discussed.)
Using the above,
we proceed to formalize the theorem implicitly in force
in~\cite{Deffayet:2023wdg,Deffayet:2021nnt},
which is a restatement of theorem 5 in
the second version of~\cite{Heredia:2024wbu}.

\begin{theorem}
\label{th:coercive}
Given a system $\emph{Sol}$
with $\mathcal{P}=\mathbb{R}^n$,
if there exists a conserved quantity
$Q:\mathbb{R}^n\rightarrow \mathbb{R}$
that is coercive,
then the system is \emph{G1}.
\end{theorem}

As in theorem \ref{th:energy} before,
G1 refers to
both equivalent definitions \ref{def:G1bounded} and \ref{def:G1compact}.
The converse of theorem \ref{th:coercive} is not true;
a counterexample being
the renowned Kronecker foliation of the torus
with irrational slope\footnote{
This consists of
curves winding on the torus,
the image of each curve being dense on the torus.
As such,
the closure of each curve
is the whole torus,
which is compact.
Thus,
the system is G1 in the compact sense of definition \ref{def:G1compact}.
However,
the system does not have a non-trivial conserved quantity
in the sense of definition \ref{def:consquant}.
To see this,
recall that a conserved quantity has constant value
on the image of each fixed curve.
As noted,
this is a dense subset on the torus.
For a continuous conserved quantity,
this entails a constant value on the whole torus,
rendering the conserved quantity trivial.
\label{foo:torus}},
see e.g.~19.18 in~\cite{Lee}.

The proof is relegated to appendix \ref{app:thenergy}.
It is worth highlighting
that theorem \ref{th:coercive}
is a corollary of theorem \ref{th:conf} later on,
whose proof is deferred to section \ref{sec:confining}.

A weighty remark is due.
Given a fixed dynamical system,
the unveiling of
conserved quantities is widely recognized
to be a non-trivial task.
This recalcitrant open problem
has recently attracted the efforts of
the machine learning community,
with noteworthy success in specific settings~\cite{Liu:2020omw,Ha:2021},
including the symbolic (i.e.~analytic) characterization of
some of the thusly disclosed conserved quantities~\cite{Liu:2022}.
Furthermore,
given a set of (non-trivial, functionally independent)
conserved quantities for a system,
it is arbitrarily difficult to prove their exhaustiveness.
Arnold-Kozlov-Neishtadt insightfully revisit six approaches
to tackle such an ambitious goal~\cite{AKN}.

In the ghost system literature,
there exist four proposals to surmount
the hardships of conserved quantity revelation.
Smilga has put forward
a systematic variational approach
to generate multi-variable conserved ghost Hamiltonians
starting from a single-variable conserved Hamiltonian, see~\cite{Damour:2021fva} and (self-)references therein.
Additionally,
he suggested~\cite{Smilga:2021nnx}
appropriate modifications
on the Korteweg-de Vries (KdV) system~\cite{Korteweg:1895lrm}
for its ghostly conversion,
later widened in~\cite{Fring:2024brg}.
Kaparulin and coauthors have developed
a methodology to extend their targeted bi-Hamiltonian theories
in a manner that the original conserved quantities
are preserved~\cite{Kaparulin:2015owa,Kaparulin:2014vpa,Kaparulin:2017uar}.
Lastly,
Deffayet et al.
consider integrable ghost systems
resulting from the application of complex canonical transformations
on standard integrable Hamiltonian systems~\cite{Deffayet:2023wdg}.
Certainly,
one may always attempt to pull through direct computation,
especially in propitious settings
and equipped with a simplifying ansatz~\cite{Heredia:2024wbu},
as we do to obtain (\ref{eq:exQ}).

\subsection{Main result: generalization of the method}
\label{sec:confining}

In light of definitions \ref{def:G1compact} and \ref{def:consquant},
critical examination of theorems \ref{th:energy} and \ref{th:coercive}
leads to the following fertile conclusion.
Consider a system (\ref{eq:soldef})
possessing a conserved quantity $Q$.
For such a system to be G1 in the compact sense,
it is sufficient that each path-connected component
of a level set of $Q$ 
is relatively compact with respect to $\mathcal{P}$.

To grasp the meaning and value of our inference,
let us take a moment to examine a purposeful example.
Let $\mathcal{P}=\mathbb{R}^2$
and let $\textrm{Sol}$ be a system of the form (\ref{eq:soldef})
with the conserved quantity
\begin{align}
\label{eq:toyQ}
Q=\sin(u^2+v^2).
\end{align}
The level set
\begin{align}
\label{eq:toyLS}
Q^\neg(0)=
\big\{\, (u,v)\in\mathbb{R}^2 \, | \,
u^2+v^2= k\pi, \,  k\in\mathbb{N}\cup\{0\} \, \big\}
\end{align}
consists of concentric circles in $\mathbb{R}^2$.
Each of these circles
is a path-connected component of $Q^\neg(0)$.
In fact,
for each $c\in[-1,1]$,
the level set $Q^\neg(c)$
is made up of concentric circles.
Considered together,
the set of level sets $\{Q^\neg(c)\}$ covers $\mathbb{R}^2$.
Now,
since all curves $\gamma\in\textrm{Sol}$ are continuous,
the image of each such curve $\gamma(I_\gamma)$
can only be contained
within one of the circles (i.e.~path-connected components) of the
appropriate level set $Q^\neg(c)$.
Even though $Q^\neg(c)$
is not relatively compact with respect to $\mathbb{R}^2$
(it contains points with arbitrarily large modulus),
the contemplated system is G1 in the compact sense
because each path-connected component of $Q^\neg(c)$
is compact
--- thus, relatively compact
--- with respect to $\mathbb{R}^2$.

The above indicated and illustrated
sufficient criterion
for the enforcement of definition \ref{def:G1compact}
on a system $\textrm{Sol}$
motivates our putting forward the subsequent definition.

\begin{definition}
\label{def:conf}
Let $\mathcal{P}$ be a Hausdorff space
and let $\mathcal{S}$ be a set.
A function $f:\mathcal{P}\rightarrow \mathcal{S}$
is \textbf{confining}
if the path-connected components of its level sets
are relatively compact with respect to $\mathcal{P}$.
\end{definition}
To the best of our knowledge, 
the above is a novel concept. 
It generalizes
coercive (proper) functions,
in the appropriate setting,
as shown 
in proposition \ref{prop:coerconf} (\ref{lem:coerconf})
of appendix \ref{app:thenergy} (\ref{app:proper}).
Its interest is justified by
our main result:

\begin{theorem}
\label{th:conf}
Given a system $\emph{Sol}$ with $\mathcal{P}$ Hausdorff,
if there exists a conserved quantity $Q:\mathcal{P}\rightarrow N$
that is confining,
then the system is \emph{G1} in the compact sense.
\end{theorem}

Recall that G1 in the compact sense was defined in \ref{def:G1compact}.
Notice that definition \ref{def:consquant}
entails that $N$ in theorem \ref{th:conf}
must be some topological space.
The converse of theorem \ref{th:conf} is not true.
A counterexample is again provided by
the Kronecker foliation of the torus
with irrational slope,
see e.g.~19.18 in~\cite{Lee}
and footnote \ref{foo:torus}.

Our theorem \ref{th:conf} can be demonstrated as follows. 

\begin{proof}
\label{pr:conf}
Consider an arbitrary but fixed curve $\gamma\in \text{Sol}$
and a point $t_0\in I_\gamma$.
Let $q_0\equiv Q(\gamma(t_0))$.
Since $Q$ is a conserved quantity,
for all $t\in I_\gamma$,
the equality $Q(\gamma(t))=q_0$ holds true.
Therefore,
the image of the curve $\gamma$
is contained in a level set of $q_0$:
$\gamma(I_\gamma)\subset Q^{\neg}(q_0)$.
Moreover,
because the curve $\gamma$ is continuous,
it must be inside a path-connected component of $Q^{\neg}(q_0)$.
Let $A$ denote such a component.
Since $Q$ is confining,
$A$ is relatively compact with respect to $\mathcal{P}$.
Hence,
its closure $\overline{A}$ is compact.
Consequently,
$\gamma(I_\gamma)$ is inside a compact set.
As $\mathcal{P}$ is Hausdorff,
$\gamma(I_\gamma)$ is relatively compact.
\end{proof}

We stress that our theorem \ref{th:conf}
provides a generalization to theorem \ref{th:coercive},
which in turn extends theorem \ref{th:energy}.
The concatenation is vertebrated by the following 
doublet of implications.
First,
a function $f:\mathbb{R}^n\rightarrow \mathbb{R}$
that is twice differentiable
satisfies
\begin{align}
\label{eq:impl1}
f \text{ has a minimum and positive definite Hessian everywhere }
\Rightarrow
\,\,
f \text{ has a minimum and is coercive. }
\end{align}
Second,
a function $f:\mathbb{R}^n\rightarrow \mathbb{R}$
fulfills
\begin{align}
\label{eq:impl2}
f \text{ is coercive }
\Rightarrow
\,\,
f \text{ is confining. }
\end{align}
The first (second) implication is proven
in lemma \ref{lem:defposcoercive} (proposition \ref{prop:coerconf})
of the appendix.
The converses are not true.
Ghostly counterexamples
can be found in the next section \ref{sec:ex}:
see (\ref{eq:exth2no1})
and (\ref{eq:exth3no2}),
respectively. 
Also, 
see figure~\ref{fig:plotexs}.

Before advancing to concrete families of ghost systems
that manifest the non-triviality of
the stepwise progression (\ref{eq:impl1})-(\ref{eq:impl2}),
we comment on an often overlooked
yet physically germane relation:
that between the protagonist conserved quantity $Q$
and the symmetries of the system $\textrm{Sol}$.

\subsection{Relation to symmetries}
\label{sec:symmetry}

All methods known to date
that prove G1 stability
(in the sense of definition \ref{def:G1compact})
in a given ghost system
use a conserved quantity.
The symmetries associated with such G1-certifying conserved quantities 
are generically not natural.
This is true
for the examples in~\cite{Deffayet:2023wdg,Heredia:2024wbu,Deffayet:2021nnt},
as well as our examples in the subsequent section \ref{sec:ex}.
Given that unnatural symmetries 
come about sporadically
in the physics literature,
we here provide 
a lightning introduction,
based on the all-encompassing recent book~\cite{Munoz-Lecanda:2024uyr}
on the geometry of mechanical systems.

A dynamical symmetry of a system $\textrm{Sol}$
is defined as
a diffeomorphism $\Psi: \mathcal{P} \rightarrow \mathcal{P}$
mapping solutions to solutions.
The subset of dynamical symmetries
related to an (at least $C^2$) conserved quantity via Noether's theorem
are called Noether symmetries. 
For regular Lagrangian systems,
such as all aforementioned ones,
a Noether symmetry leaves 
both the Lagrangian energy and the symplectic form invariant,
but not necessarily the Lagrangian.

On the other hand,
let $\mathcal{P}$ be the tangent space of the configuration space $\mathcal{P}=T\mathcal{C}$.
A diffeomorphism $\Psi:T\mathcal{C}\rightarrow T\mathcal{C}$ is said to be natural if
it consists of a transformation in positions
that is canonically lifted to velocities:
\begin{align}
\label{eq:unnat}
\Psi(\dot{x}^i,x^i)=
(\dot{x}^{\prime i}
(\dot{x}^i,x^i),x^{\prime i}(\dot{x}^i,x^i)),
\qquad
\textrm{ such that } \quad
x^{\prime i}(\dot{x}^i,x^i)=x^{\prime i}(x^i)
\quad\textrm{ and } \quad
{\dot{x}}^{\prime i}=
\frac{\partial x^{\prime i}}{\partial x^j}\dot{x}^j.
\end{align}
A Noether symmetry with the naturalness property leaves
the Lagrangian invariant (up to a divergence).
The converse is also true:
a natural symmetry that leaves the Lagrangian invariant (up to a divergence)
is a natural Noether symmetry.

It follows that
the Noether symmetries of the G1-enforcing conserved quantity
in ghost systems
typically present two unconventionalities.
First,
they entail a transformation of positions involving velocities.
Second,
they may,
or may not,
leave the Lagrangian invariant (up to a divergence).

The quintessential example of unnatural symmetries is
given by the Laplace-Runge-Lenz vector in the Kepler problem.
(Re)Discovered several times,
it acquired sizeable notoriety in the early days of quantum mechanics,
when Pauli employed it to derive the energy levels of the hydrogen atom~\cite{Pauli:1926qpj}.
Ulterior works showed that the symmetries associated to the Laplace-Runge-Lenz vector,
together with the angular momentum,
form a Lie group acting on the phase space taken as a whole
--- i.e.~considering both positions and momenta.
The interested reader is referred to
the enjoyable text~\cite{Baez} 
and references therein.

\paragraph{Relation to G2 and Liouville(-Mineur)-Arnold theorem.}
\label{par:relationLAT}
Two reminders and a fruitful application follow.
First,
proving that a system $\textrm{Sol}$ is G1
in the sense of definition \ref{def:G1compact} 
proves that the system is G2 in the sense of definition \ref{def:G2formal}
--- recall figure~\ref{fig:G1s}.
Second,
if the system $\textrm{Sol}$ can be described by
an (at least $C^1$) vector field $X$, 
then G2 stability is equivalent to the completeness of $X$
--- summon back the comments below definition \ref{def:G2formal}
and observe this applies to
all examples in~\cite{Deffayet:2023wdg,Heredia:2024wbu,Deffayet:2021nnt}
and section \ref{sec:ex}.
In this case,
G1 implies the completeness of $X$. This direct consequence 
is constructive in the context of integrable systems. 

We already stressed that finding conserved quantities is generically difficult.
Consider nonetheless the auspicious scenario
wherein a dynamical system $\textrm{Sol}$
possesses at least $n=\textrm{dim}(\mathcal{C})$ 
functionally independent conserved quantities $\{Q\}$ in involution.
Such a system is said to be {\it integrable}.
The study of integrable systems
is a decidedly active area of research on its own~\cite{Miranda},
often invoking supplementary non-trivial hypotheses
from the onset.
This has led to a notable ramification within the field
--- for a first impression on the bustle,
see the start of chapter 2.2.3
in the wonderful book~\cite{KAMStory}.
In our fleeting incursion into the subject,
we opt for following~\cite{Sardana}
in our understanding of the vertebral Liouville-Arnold theorem~\cite{Lio,Arnold} ---
independently established by Mineur in the prolongedly
overlooked works~\cite{Mineur1,Mineur2}, which have been
fortunately resurfaced in the last decade.
We thus require that
the infinitesimal symmetries associated to the conserved quantities
be complete vector fields.

In an integrable mechanical system
proven to be G1 by means of our theorem \ref{th:conf},
at least one of the conserved quantities in $\{Q\}$ is confining.
As a result,
the hypotheses of the Liouville-Arnold theorem are satisfied.
Moreover,
the system can be described as
free particles moving {\it exclusively} around torii.
Thusly,
theorem \ref{th:conf} connects stable ghost systems 
with the forefront of the field of integrable systems,
see~\cite{Miranda} and references therein.
The particulars and proof
of this relation are allocated to appendix \ref{app:integrability}.

\section{Examples}
\label{sec:ex}

This section is devoted to illustrating the impactful
implications of our previous results,
especially definition \ref{def:G1compact}
and theorem \ref{th:conf},
in the particular arena of ghost systems.
To achieve this goal,
we begin by presenting a two free parameter and one free function
family of integrable ghost systems (\ref{eq:exLag}).
The subsequent,
staggered demand that the axioms of theorems \ref{th:energy}, \ref{th:coercive} and \ref{th:conf}
be fulfilled by (\ref{eq:exLag}) yields
increasingly vast ghost systems
proven to be G1 in the compact sense of definition \ref{def:G1compact}.
Most remarkably,
the G1 ghost systems obtained via the enforcement of theorem \ref{th:conf}
are not only new,
but also could not have been derived
with the techniques in the literature.

Consider the regular Lagrangian system
\begin{align}
\label{eq:exLag}
L=
\frac{1}{2}(\dot{x}^2-\dot{y}^2)
-\frac{\omega^2}{2}(x^2-y^2)
-f(x+\alpha y),
\end{align}
where dot denotes derivation with respect to time,
the function $f:\mathbb{R}\rightarrow\mathbb{R}$
is taken to be at least twice differentiable
and $\omega^2,\alpha\in\mathbb{R}$.
We highlight that (\ref{eq:exLag})
is not an instance of a self-interacting extension of
the abundantly studied Pais-Uhlenbeck oscillator~\cite{Kaparulin:2015pda,Pavsic:2013noa,Ilhan:2013xe,Bender:2008gh,Mostafazadeh:2010yw},
owing to the degeneracy
in the frequencies of the independent variables $(x,y)$.

Because (\ref{eq:exLag}) is autonomous,
the (Lagrangian) energy
\begin{align}
\label{eq:exE}
E_L=
\frac{1}{2}(\dot{x}^2-\dot{y}^2)
+\frac{\omega^2}{2}(x^2-y^2)
+f(x+\alpha y)
\end{align}
is an obvious conserved quantity of the system.
Furthermore,
it is apparent that (\ref{eq:exLag}) is ghostly:
the energy (\ref{eq:exE}) does not have a global minimum.

It can be readily verified
that the system (\ref{eq:exLag})
has the functionally independent conserved quantity
\begin{align}
\label{eq:exQ}
Q=
\frac{1}{2}\left(\dot{x}^2+\dot{y}^2+2c_1\dot{x}\dot{y}\right)
+\frac{\omega^2}{2}\left(x^2+y^2+2c_1xy\right)
+c_2f(x+\alpha y),
\end{align}
with\footnote{The shrewd reader
may prefer to let $\alpha=\tan\theta$,
so that $c_1=\sin(2\theta)$
and $c_2=\cos(2\theta)$.}
\begin{align}
\label{eq:excons}
c_1=
\frac{2\alpha}{1+\alpha^2},
\qquad
c_2=
\frac{1-\alpha^2}{1+\alpha^2}.
\end{align}
We have unveiled (\ref{eq:exQ})
by direct computation,
in a befitting adaptation of
the classical results by Darboux~\cite{Darboux}.

To the best of our knowledge,
(\ref{eq:exLag}) with (\ref{eq:exQ})
is a novel family of integrable ghost systems.
In particular,
it cannot be obtained
via complex canonical transformations
on previously known systems,
as suggested in~\cite{Deffayet:2023wdg}.
Besides,
it is neither part of the family in equation (34) of~\cite{Kaparulin:2014vpa}
nor contained in any of
the families I-IV in section 3.4 of~\cite{Heredia:2024wbu}.

Further properties of (\ref{eq:exLag})
are discussed in appendix~\ref{app:properties},
including symmetries and its bi-Hamiltonian nature.

As announced,
we proceed to carving out various cases from (\ref{eq:exLag})
that are G1 in the compact sense of definition \ref{def:G1compact}.
To this aim,
we shall employ the results in section \ref{sec:Qs}.

\subsection{An application of theorem \ref{th:energy}}
\label{sec:ex1}

In order to determine the instances in which the system (\ref{eq:exLag})
is G1 in the equivalent sense of definitions \ref{def:G1bounded}
and \ref{def:G1compact},
we here employ theorem \ref{th:energy}.
We apply it,
not on the energy (\ref{eq:exE})
that always fails to meet the hypotheses,
but rather on the conserved quantity (\ref{eq:exQ}).
Specifically,
we require that (\ref{eq:exQ})
has a positive definite Hessian everywhere,
as well as a minimum.

For starters,
we compute the Hessian of (\ref{eq:exQ}),
which is of block-diagonal form
\begin{align}
\label{eq:hessQ}
\textrm{Hess}(Q)=
\begin{pmatrix}
H_1 & 0 \\
0 & H_2
\end{pmatrix},
\qquad
\textrm{where }
\quad 
H_1=
\begin{pmatrix}
1 & c_1 \\
c_1 & 1 
\end{pmatrix},
\qquad
H_2=
\omega^2 H_1+
c_2 f^{\prime\prime}
\begin{pmatrix}
1 & \alpha \\
\alpha & \alpha^2 \\
\end{pmatrix}.
\end{align}
Here,
the prime stands for total derivative:
$f^\prime(z)\equiv df/dz$ and so on.
As per Sylvester's criterion,
(\ref{eq:hessQ}) is positive definite iff all its principal minors
have a positive determinant.
The first principal minor is evidently positive.
For the second principal minor $H_1$ to have positive determinant,
we require that
\begin{align}
\label{eq:minor2}
1-c_1^2>0
\,\, \Rightarrow \,\,
(\alpha^2-1)^2>0
\,\, \Rightarrow \,\,
\alpha\neq\pm1.
\end{align}
Provided that (\ref{eq:minor2}) is fulfilled,
the positiveness of the determinant of
the third principal minor entails
\begin{align}
\label{eq:minor3}
(1+\alpha^2)\omega^2+(1-\alpha^2)f^{\prime\prime}>0.
\end{align}
Besides,
the determinant of (\ref{eq:hessQ})
is positive when
\begin{align}
\label{eq:minor4}
(1-c_1^2)\omega^4
+c_2\omega^2(\alpha^2-2c_1\alpha+1)f^{\prime\prime}>0
\,\, \Rightarrow \,\,
\omega^2\Big(\omega^2 +(1-\alpha^2)f^{\prime\prime}\Big)>0.
\end{align}
Minor algebraic effort is involved
in exhaustively choosing the tightest bound among
(\ref{eq:minor3}) and (\ref{eq:minor4}).
We thus arrive at the necessary and sufficient conditions
for the positive definiteness of (\ref{eq:hessQ}):
\begin{align}
\label{eq:poscond}
\omega^2>0,
\qquad
f^{\prime\prime} \gtrless \frac{\omega^2}{\alpha^2-1}
\quad \textrm{for } |\alpha|\lessgtr 1,
\end{align}
with $|\alpha|=1$ forbidden by (\ref{eq:minor2}).
We call attention to the fact that
we have derived $\omega^2>0$,
as opposed to taking it as a priori datum.

On the other hand,
for (\ref{eq:exQ}) to have a minimum,
we demand
\begin{align}
\label{eq:derQ0}
\frac{\partial Q}{\partial \dot{x}}=0,
\qquad
\frac{\partial Q}{\partial \dot{y}}=0,
\qquad
\frac{\partial Q}{\partial x}=0,
\qquad
\frac{\partial Q}{\partial y}=0.
\end{align}
The first two requirements immediately yield $\dot{x}=0=\dot{y}$.
Operating through,
the latter two requirements give rise to
an additional, necessary and sufficient condition
for the existence of a minimum:
that there exists a point $z_0\in\mathbb{R}$ such that
\begin{align}
\label{eq:mincond}
\omega^2z_0+(1-\alpha^2)f'(z_0)=0.
\end{align}
If extant,
such $z_0$ can be readily employed to
determine the point at which
(\ref{eq:exQ}) achieves its minimum:
\begin{align}
\label{eq:minvalue}
x=-\frac{1}{\omega^2}f^{\prime}(z_0),
\qquad
y=\frac{\alpha}{\omega^2}f^{\prime}(z_0).
\end{align}

All in all,
(\ref{eq:poscond}) and (\ref{eq:mincond})
comprise the necessary and sufficient conditions
for $Q$ in (\ref{eq:exQ})
to have a positive definite Hessian everywhere and a minimum.

\paragraph{Cases in point.}
Let $n\in\mathbb{N}$ and let $a,b\in\mathbb{R}$.
Exemplars from the family (\ref{eq:exLag})
satisfying both (\ref{eq:poscond}) and (\ref{eq:mincond})
include
\begin{align}
\label{eq:exth1}
f(z)=z^{2n}
\,\, \textrm{ for } \,\, 
|\alpha|<1,
\qquad
f(z)=az+b,
\qquad
f(z)=(z^2+1)^{1/n},
\end{align}

As around (\ref{eq:Hcounterth1}) before,
we again call attention to the fact that the simultaneous fulfillment of
(\ref{eq:poscond}) and (\ref{eq:mincond}) is needed
in order to prove G1 stability of a system of the form (\ref{eq:exLag}).
Counterexamples are
\begin{align}
\label{eq:counterth1}
f(z)=z^{2n}
\quad \textrm{ for } \,\, 
|\alpha|>1,
\qquad
f(z)=\frac{\omega^2}{1-\alpha^2}\big(\textrm{exp}(z)-z^2/2\big)
\quad \, \forall \,\, 
\alpha\neq\pm1.
\end{align}
These fail to meet (\ref{eq:poscond}) and (\ref{eq:mincond}), respectively.

An altogether distinct example
of a ghost system that is proven to be G1
as per theorem \ref{th:energy} appeared not long ago,
in section 3.3 of~\cite{Heredia:2024wbu}.

\subsection{An application of theorem \ref{th:coercive}}
\label{sec:ex2}

Next,
we ascertain the subset of (\ref{eq:exLag})
that is G1 in the equivalent sense of definitions \ref{def:G1bounded}
and \ref{def:G1compact}
with a basis on theorem \ref{th:coercive}.
Consequently,
we move forward to imposing that
the conserved quantity $Q$ in (\ref{eq:exQ})
is coercive.

In fact,
this is easier said than done.
Definition \ref{def:limcoerc} of a coercive function
is a usual suspect for finding counterexamples,
as opposed to a powerful hammer for forging a proof.
Accordingly,
we start by
restricting the free parameter space of $Q$,
spotting escape directions that are incompatible with
its desired coerciveness.
First,
we note that,
for $\alpha=\pm1$,
points of the form $(\dot{x},\dot{y},x,y)=(\lambda,\mp\lambda,0,0)\equiv\mathbf{p}_\lambda$,
with $\lambda\in\mathbb{R}$
satisfy
\begin{align}
\label{eq:zerolim}
\lim_{\lambda\rightarrow \pm\infty}
|Q({\mathbf{p}_\lambda})|=|c_2f(0)|\neq+\infty,
\end{align}
The above contradicts definition \ref{def:limcoerc} of a coercive function.
Thence,
in analogy to the preceding section \ref{sec:ex2},
we hereafter restrict attention to systems where $\alpha\neq \pm1$.
Second,
consider 
\begin{align}
\label{eq:wposQ}
\mathbf{q}_\lambda\equiv
\lambda(-\alpha\sqrt{|\omega^2|},\sqrt{|\omega^2|},-\alpha,1),
\qquad
Q({\mathbf{q}_\lambda})
=c_2(1-\alpha^2)\Big(\frac{|\omega^2|+\omega^2}{2}\Big)\lambda^2+c_2f(0).
\end{align}

If $\omega^2\leq0$,
it is clear that
\begin{align}
\label{eq:zerolim2}
\lim_{\lambda\rightarrow \pm\infty}
Q({\mathbf{q}_\lambda})=c_2f(0)\neq+\infty.
\end{align}
We point out the restrictive nature of (\ref{eq:zerolim2})
as compared to definition \ref{def:limcoerc}.
In detail,
we have dropped out the absolute value,
thereby forcing $Q({\mathbf{q}_\lambda})$ to tend to {\it plus} infinity
as $\lambda$ diverges.
Indeed,
this is the only possibility for $Q$ in (\ref{eq:exQ}) to be coercive,
since its kinetic piece is invariably positive definite with a minimum
for all $\alpha\neq\pm1$
--- recall (\ref{eq:minor2}) and around.
(This concept is formally introduced 
in appendix \ref{app:thenergy}:
see definition \ref{def:rescoer} of a {\it positively coercive} function.)
Thus, (\ref{eq:zerolim2}) contradicts definition \ref{def:limcoerc}.
As a result,
once more paralleling the preceding section \ref{sec:ex1},
we derive the requirement that frequencies be real for $Q$ to be coercive:
$\omega^2>0$.

In the same spirit,
we move forward to restricting the free functional space of $Q$.
Let the velocities $(\dot{x},\dot{y})$ take a fixed value.
In this case, 
the enforcement of definition \ref{def:rescoer} on $Q$ stipulates that
\begin{align}
\label{eq:potQ}
U(x,y)=
\frac{\omega^2}{2}\left(x^2+y^2+2c_1xy\right)
+c_2f(x+\alpha y)
\end{align}
be positively coercive.
It is beneficial to perform the change of variables
\begin{align}
\label{eq:newcoords}
\begin{cases}
\zeta=x+\alpha y,
\\
\eta=\alpha x+ y,
\end{cases}
\end{align}
so that
\begin{align}
\label{eq:potQnew}
U(\zeta,\eta)=\frac{\omega^2}{2(1+\alpha^2)}(\zeta^2+\eta^2)+c_2f(\zeta)=P(\zeta,\eta)+c_2f(\zeta).
\end{align}
Let $\eta$ take a fixed value.
For later convenience, 
we single out $\eta=0$.
Then,
for $Q$ to be positively coercive,
it must hold true that
\begin{align}
\label{eq:limzeta}
\lim_{\zeta\rightarrow\pm\infty} P(\zeta,0)+c_2f(\zeta)=+\infty,
\end{align}
which readily leads to $f(\zeta)$'s of the form
\begin{align}
\label{eq:fform}
c_2f(\zeta)=h(\zeta)-P(\zeta,0),
\end{align}
for any $C^2$ positively coercive function $h:\mathbb{R}\rightarrow\mathbb{R}$.
(That $h$ be $C^2$,
as opposed to simply continuous,
follows from the Lagrangian nature of the system (\ref{eq:exLag}).)

The structure (\ref{eq:fform}) is not only necessary,
but also sufficient for $Q$ to be coercive.
To see this,
we need only to prove that the potential piece $P+c_2f$ of $Q$ is positively coercive
--- the kinetic piece always being positive definite with minimum.
With this goal in mind and as per definition \ref{def:rescoer},
we must ensure that 
the multivariate limit to infinity of $P+c_2f$ diverges.
Namely,
we must find,
for every $k>0$,
a radius $R>0$ such that
$P+c_2f$ takes a value larger than $k$
at any point with modulus larger than $R$.

Because $h$ is positively coercive,
for any $k>0$
there exist a $\zeta_0\in\mathbb{R}$ such that,
for all $|\zeta|>\zeta_0$,
\begin{align}
\label{eq:lboundP}
P(\zeta,0)+c_2f(\zeta)>k.
\end{align}
Since $[-\zeta_0,\zeta_0]$ is compact with respect to $\mathbb{R}$
and $f$ is continuous,
$c_2f$ reaches a global minimum when evaluated in this interval.
Let $m$ denote such minimum.
Notice that,
by analogy to the kinetic piece,
$P$ has positive definite Hessian everywhere
and a minimum at $\zeta=0=\eta$.
This implies that
there exists $R>0$ such that,
for all $\zeta,\eta$ satisfying $\|(\zeta,\eta)\|>R$,
it holds true that $P(\zeta,\eta)>k-m$.
This completes the proof,
owing to the fact that,
for $\|(\zeta,\eta)\|>R$,
either $|\zeta|>\zeta_0$ and
\begin{align}
\label{eq:Pbound1}
P(\zeta,\eta)+c_2f(\zeta)>P(\zeta,0)+c_2f(\zeta)>k,
\end{align}
or $\zeta\in[-\zeta_0,\zeta_0]$,
in which case
\begin{align}
\label{eq:Pbound2}
P(\zeta,\eta)+c_2f(\zeta)>k-m+m=k.
\end{align}

Summing up,
for $Q$ to be coercive,
it is necessary and sufficient that
$|\alpha|\neq1$, $\omega^2>0$ and (\ref{eq:fform}) is fulfilled.
We draw attention to the fact that
these requirements encompass the previously derived
necessary and sufficient conditions (\ref{eq:poscond}) and (\ref{eq:mincond})
for $Q$ to have a positive definite Hessian everywhere and a minimum.
In other words,
the systems of the form (\ref{eq:exLag})
that are G1 on the grounds of theorem \ref{th:coercive}
include and extend those that are G1 on the grounds of theorem \ref{th:energy}.

\paragraph{Cases in point.}
Let $a,b\in\mathbb{R}$.
Concrete cases from the family (\ref{eq:exLag})
with a conserved quantity $Q$ of the form (\ref{eq:exQ})
that is coercive include:
all previous examples with a $Q$ that has positive definite Hessian everywhere and minimum (\ref{eq:exth1})
and instances beyond,
such as
\begin{align}
\label{eq:exth2no1}
f(z)=a(z^2-1)(z^2-4) 
\,\, \textrm{ for } \,\, 
a\gtrless0
\,\, \textrm{ and } \,\, 
|\alpha|\lessgtr 1,
\qquad
f(z)=b \sin(z).
\end{align}
These proposals' $Q$s
only have a Hessian that is positive definite everywhere
when
\begin{align}
    0\lesseqgtr a\lesseqgtr\frac{\omega^2}{10(1-\alpha^2)}
\,\, \textrm{ and } \,\, 
|\alpha|\lessgtr 1, \qquad\qquad  |b|< \frac{\omega^2}{|1-\alpha^2|},
\end{align}
respectively. 

Independent ghost systems formerly proven to be G1
on the (implicit) grounds of theorem \ref{th:coercive}
can be found in~\cite{Deffayet:2023wdg,Deffayet:2021nnt}
and in section 3.4 of~\cite{Heredia:2024wbu}.

\subsection{An application of theorem \ref{th:conf} }
\label{sec:ex3}

At last,
we establish the subset of (\ref{eq:exLag}) that is G1
in the equivalent sense of definitions \ref{def:G1bounded} and \ref{def:G1compact}
as per theorem \ref{th:conf}.
Namely,
we now require that the conserved quantity (\ref{eq:exQ}) is confining ---
recall our proposed definition \ref{def:conf}.

The same reasoning as that
between equations (\ref{eq:zerolim}) and (\ref{eq:zerolim2}) earlier on
applies in this case.
Indeed,
our discussion shows the existence of a line inside a level set of $Q$.
In particular,
this line is inside a path-connected component of a level set of $Q$.
This contradicts definition \ref{def:conf}.
Thus,
for $Q$ to be confining,
it must hold true that
$\alpha\neq\pm1$ and $\omega^2>0$.

Having restricted the free parameter space in (\ref{eq:exLag})
so as to make it compatible with a confining $Q$ in (\ref{eq:exQ}),
we replicate the strategy in the previous section \ref{sec:ex2}
and turn to restricting the free functional space.
To this aim,
we reconsider (\ref{eq:fform}).
Specifically,
we ponder over relaxations of the (sufficiently smooth) function $h(\zeta)$.
In seeking to exceed the positively coercive requirement,
we envisage an oscillatory function
whose maxima ultimately increase
while minima remain below some threshold.
For instance,
see plot D in figure~\ref{fig:plotexs}.

Rigorously,
let $f(\zeta)$ be of the form
\begin{align}
\label{eq:fformconfining}
c_2f(\zeta)=h(\zeta)-P(\zeta,0),
\end{align}
where $h:\mathbb{R}\rightarrow \mathbb{R}$
is any $C^2$ function such that,
for all $k\in\mathbb{R}$ and for all $\zeta\in\mathbb{R}$,
there exists $\zeta_1,\,\zeta_2\in\mathbb{R}$
with $\zeta_1<\zeta<\zeta_2$
such that $f(\zeta_1)>k$ and $f(\zeta_2)>k$.
Then,
(\ref{eq:exQ}) can be conveniently rewritten as
\begin{align}
\label{eq:Qzetaeta}
Q=
\frac{1}{2}\left(\dot{x}^2+\dot{y}^2+2c_1\dot{x}\dot{y}\right)
+P(\zeta,\eta)+h(\zeta)-P(\zeta,0)=Q_1+h(\zeta)=Q_2+h(\zeta)-P(\zeta,0).
\end{align}
Subsequently,
we prove that (\ref{eq:Qzetaeta}) is confining iff (\ref{eq:fformconfining}) is satisfied.

Consider the level set $Q^{\neg}(c)$
for some arbitrary $c\in\mathbb{R}$.
Let $A\subset\mathbb{R}^4$ denote a path-connected component thereof.
Further,
let $J\subset\mathbb{R}$ denote the projection of $A$ to the $\zeta$ variable.
First, 
we establish that $J$ is relatively compact with respect to $\mathbb{R}$.
Afterwards,
we demonstrate that $A$ is relatively compact with respect to $\mathbb{R}^4$.
This entails the sufficiency of (\ref{eq:fformconfining}) for
rendering (\ref{eq:Qzetaeta}) confining.

By virtue of (\ref{eq:fformconfining}),
for any point $\zeta_0\in J$  and choosing $k=c$,
there exist $\zeta_1<\zeta_0<\zeta_2$
such that $h(\zeta_1)>c$ and $h(\zeta_2)>c$.
Within $Q^{\neg}(c)$,
the values of $h(\zeta)$ are contained in $(-\infty,c]$,
since zero is the absolute minimum value of $Q_1$.
Consequently,
such $\zeta_1$ and $\zeta_2$ do not belong in $J$.
Furthermore, 
$J$ is path-connected because it is the image of the path-connected set $A$ by the projection,
which is a continuous map.
Altogether,
$J$ is a subset of $[\zeta_1,\zeta_2]$;
thus,
$J$ is relatively compact with respect to $\mathbb{R}$. 

The continuity of $h(\zeta)-P(\zeta,0)$
warrants its attainment of a global minimum value $m$
on the closure of $J$.
Consistently,
the maximum value of $Q_2$ on $A$ is $c-m$.
In fact,
$Q_2$ is a positive elliptic paraboloid
and, as such,
positively coercive as a function of its variables
--- which becomes apparent in terms of $(\dot{x},\dot{y},x,y)$.
Definition \ref{def:limcoerc},
in the restrictive sense wherein the absolute value is removed,
assures us of the existence of a radius $R$ such that
$Q_2(\dot{x},\dot{y},x,y)>c-m$
whenever $\|(\dot{x},\dot{y},x,y)\|>R$.
It follows that $A\subset B(0,R)$;
thereby,
$A$ is relatively compact with respect to $\mathbb{R}^4$.

To wrap up the proof,
we show by contradiction that (\ref{eq:fformconfining})
is not only sufficient but also necessary a condition
for (\ref{eq:Qzetaeta}) to be confining.
Assume $h(\zeta)$ does not have the property (\ref{eq:fformconfining}).
That is, there exist $k\in\mathbb{R}$ and $\zeta_0\in\mathbb{R}$
such that
$h(\zeta)\leq k$
for all $\zeta_1<\zeta_0$
(and similarly for $\zeta_2>\zeta_0$).
Then, the level set 
$Q^\neg(k)$ contains the subset 
\begin{align}
\label{eq:exhaustive}
\big\{(\dot{x},\dot{y},\eta,\zeta)\in\mathbb{R}^4
\,|\,\dot{x}=\sqrt{2(k-h(\zeta))},\,
\dot{y}=0,\, \eta=0,\, \zeta\leq\zeta_0\big\}, 
\end{align}
which is path-connected yet not relatively compact;
so $Q$ is not confining. 

Retrospectively,
we can enhance our understanding of the successful ansatz (\ref{eq:fformconfining}).
The free function $f$ is decomposed into two pieces:
the free function $h$ possessing the crucial property that guarantees the system is G1,
plus the fixed function $P(\zeta,0)$ 
that balances the relevant term of $Q$
--- recall (\ref{eq:potQnew}).
(This intuition also holds for the coercive instance (\ref{eq:fform}),
where $P(\zeta,0)$ surfaced naturally.)
Moreover,
the decomposition unlocks the exhaustiveness of the result.
The term $P(\zeta,0)$ provides the exact cancellation of
the quadratic terms $\zeta^2$ inside the square root in (\ref{eq:exhaustive}),
rendering it well-defined
for all values of $\zeta\leq\zeta_0$.

In short,
for $Q$ to be confining,
it is necessary and sufficient that
$|\alpha|\neq1$, $\omega^2>0$ and (\ref{eq:fformconfining}) holds.
We highlight that
(\ref{eq:fformconfining}) generalizes (\ref{eq:fform}),
so that
the systems of the form (\ref{eq:exLag})
that are G1 on the grounds of theorem \ref{th:conf}
broaden those that are G1 on the grounds of theorem \ref{th:coercive}.

\paragraph{Cases in point.}
Let $a,b\in\mathbb{R}$.
Exemplars of (\ref{eq:exLag})
with a conserved quantity $Q$ of the form (\ref{eq:exQ})
that is confining include:
all previous examples with a $Q$ that has
minimum and positive definite Hessian everywhere (\ref{eq:exth1}),
those associated with a coercive $Q$ (\ref{eq:exth2no1})
and instances beyond,
such as
\begin{align}
\label{eq:exth3no2}
f(z)=ae^z\cos(z)\,\, \textrm{ for } \,\, a\neq0\qquad f(z)=bz^2 \sin(z)\,\, \textrm{ for } \,\, 
|b|\geq\frac{\omega^2}{2|1-\alpha^2|}\equiv b_{\textrm{c}}.
\end{align}
The left-most example is neither bounded from below nor from above,
regardless of the sign of $a$\footnote{This example refutes footnote 5 in \cite{Deffayet:2023wdg}: the corresponding conserved quantity $Q$ is bounded neither from below nor from above, yet it is confining and thus guarantees that the system is G1 as per both definitions \ref{def:G1bounded} and \ref{def:G1compact}.}.
The right-most example
neatly showcases the doublet of single-sided implications (\ref{eq:impl1})-(\ref{eq:impl2}).

\begin{figure}[!htbp]
\hspace{-0.7cm}
\includegraphics[scale=0.55]{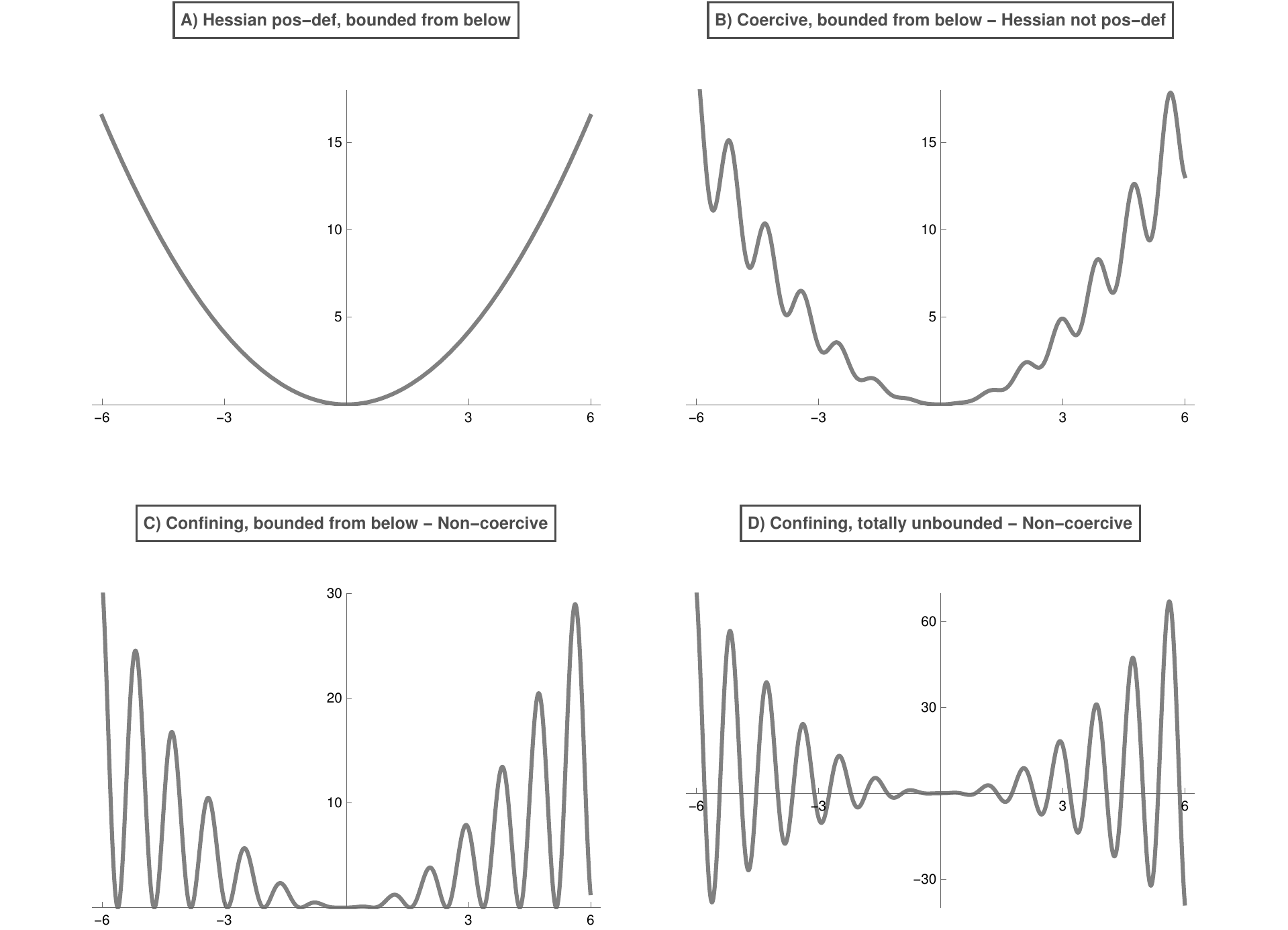}
\hspace{0.7cm}
\caption{Archetypal instances of the rich phenomenology of the right-most example in (\ref{eq:exth3no2}).
The function (\ref{eq:plottedQ}) is shown,
with its variable $\zeta$ labelling the abscissa.
The free parameters have been set to
$\omega^2=1$, $\alpha=0.3$, $\nu=7$ and
$b=0$ (plot A),
$b=1/8$ (plot B),
$b=b_c\approx0.55$ (plot C),
$b=2$ (plot D).
Thereby,
the itemization below (\ref{eq:exth3no2})
is illustrated.\\
All these are exemplary potentials 
of conserved quantities that can be used to analytically prove
G1 stability in the compact sense of definition \ref{def:G1compact}.
The lower half portrays radically new proposals,
including totally unbounded scenarios.
}
\label{fig:plotexs}
\end{figure}

\begin{itemize}
\item
For $b=0$,
the system decouples.
The conserved quantity $Q$ has positive definite Hessian everywhere and minimum.
Accordingly,
the system can be proven to be G1 through
either theorems \ref{th:energy}, \ref{th:coercive} or \ref{th:conf}.
\item
For $|b|\in(0,b_\textrm{c})$,
the conserved quantity $Q$ is coercive,
but the Hessian is not positive definite at some points.
Thence,
the system can be proven to be G1
via either theorems \ref{th:coercive} or \ref{th:conf},
but not via theorem~\ref{th:energy}.
\item
For $|b|\in[b_c,\infty)$,
the conserved quantity $Q$ is confining,
but it is neither coercive nor does it possess a positive definite Hessian everywhere.
In this case,
the system can only be proven to be G1 by means of theorem \ref{th:conf}.
\begin{itemize}
\item[$\circ$]
If $|b|=b_c$,
$Q$ has (infinitely many) global minima.
\item[$\circ$]
Else, 
$Q$ is unbounded from above and below.
\end{itemize}
\end{itemize}

Figure~\ref{fig:plotexs} exhibits the aforementioned progression in plots A to D,
respectively.
In more detail,
it represents the G1-triggering piece $h$
of the conserved quantity $Q$,
for the choice of $f$ in the right-most example of (\ref{eq:exth3no2})
and in terms of the $\zeta$ variable introduced in (\ref{eq:newcoords}):
\begin{align}
\label{eq:plottedQ}
h(\zeta)=
P(\zeta,0)+c_2f(\zeta)=
\frac{\omega^2}{2(1+\alpha^2)}\zeta^2+c_2b\zeta^2\sin(\nu\zeta).
\end{align}
(At this point,
it is beneficial to 
bring the discussion below (\ref{eq:exhaustive}) to mind.)
The parameter $\nu\in\mathbb{R}-\{0\}$
has been introduced by hand,
for an easier appreciation of the qualitative behavior of the plotted functions.

We reiterate that examples of (ghost) systems
provable to be G1 exclusively on the basis of theorem \ref{th:conf} are
novel in a deep sense.
Indeed,
they could not have been obtained
without an increase in both the understanding and techniques
employed before.

\section{Summary and conclusions}
\label{sec:concl}

On the basis of the profuse literature
in stable ghost theories,
we have put forward a cartography
for classical notions of stability currently in force.
Crucially,
we have differentiated between global and local stability notions,
depending on their application to
{\it all} or {\it some} solutions of the system:
see \ref{def:G1}-\ref{def:G2} versus \ref{def:L1}-\ref{def:L2}.
We stress that only Lyapunov-like proposals \ref{def:L2}
enjoy a long and fruitful tradition in the mathematical literature
of dynamical systems.
From this perspective,
alternative classical notions of stability enforced in ghost physics
open the door to expansive studies
in an altogether distinct discipline.

Subsequently,
we have formalized the global notions \ref{def:G1} and \ref{def:G2}:
see \ref{def:G1bounded}-\ref{def:G1compact} and \ref{def:G2formal},
respectively.
It is worth emphasizing that our propounded definitions
apply to generic dynamical systems (\ref{eq:soldef}),
independent of their possible Lagrangian/Hamiltonian and/or ghostly nature.
The fact that \ref{def:G1}
admits two formalizations \ref{def:G1bounded}-\ref{def:G1compact}
is consequential for dynamical systems
whose solutions are defined in
metric spaces without the Heine-Borel property,
as shown in figure~\ref{fig:G1s}.
We call attention to the relevance of these scenarios,
which encompass infinite-dimensional Banach and Hilbert spaces,
as well as geodesically incomplete Riemannian manifolds.
In all such cases,
there is no equivalence between \ref{def:G1bounded} and \ref{def:G1compact},
and only \ref{def:G1compact} implies \ref{def:G2formal}.
The dichotomy is to be resolved on physical grounds,
according to the desired phenomenology
of the systems under study.
For ghost systems,
we regard this as an important open question.

Taking definition \ref{def:G1compact} as bedrock throughout,
we have (re)explicated the mathematical foundations
sustaining the analytic proofs of global stability G1
for ghost systems:
see theorems \ref{th:energy} and \ref{th:coercive}.
The crux of the proofs relies on two insights.
First,
the characterization of ghost systems
as dynamical Lagrangian/Hamiltonian systems
for which the energy does not have a global minimum.
Second,
the identification of
a functionally independent conserved quantity $Q$
possessing a key property
as the means to ascertain the desired G1 feature in the system.
The key property in 
(a befitting adaptation of) theorem \ref{th:energy}
is two-fold:
$Q$ has a minimum
and its Hessian is positive definite everywhere.
A more sophisticated key property
is employed in theorem \ref{th:coercive}:
$Q$ is coercive,
in the sense of definition \ref{def:limcoerc}.
Theorems \ref{th:energy} and \ref{th:coercive}
formalize and classify the examples
in~\cite{Kaparulin:2015owa,Abakumova:2018eck,Heredia:2024wbu,Kaparulin:2014vpa,Kaparulin:2015uxa,Kaparulin:2015pda,Abakumova:2017syd,Kaparulin:2017swa,Abakumova:2017uto,Kaparulin:2017uar,Kaparulin:2018npv,Abakumova:2019ifi,Kaparulin:2019njc,Abakumova:2019wpn,Abakumova:2019dov,Kaparulin:2020gqn}
and~\cite{Deffayet:2023wdg,Heredia:2024wbu,Deffayet:2021nnt},
respectively.

We have extended the aforementioned known frameworks
in an unprecedented manner:
see our main result theorem \ref{th:conf}.
In more detail,
within the same paradigm,
we have generalized the key property to:
$Q$ is confining,
in the sense of our proposed novel definition \ref{def:conf}.
The architecture conformed by theorem \ref{th:conf} and definition \ref{def:conf}
applies to generic dynamical systems in Hausdorff spaces.
The supervened augmentation in ghost systems provable to be G1
as a direct consequence of our work
is both quantitative and qualitative.
In the latter regard,
we have explicitly proved that $Q$ need {\it not} be bounded from below
for the system to be G1,
thereby rigorously contravening the status quo of ghost physics
in a counter-intuitive way.
For an instinctive impression,
compare the upper and lower plots in figure~\ref{fig:plotexs},
remarkably the totally unbounded scenario in plot D.

We have introduced
confining functions
tailored to
the stability definition \ref{def:G1compact} for a dynamical system.
Confining functions englobe coercive functions 
in metric spaces with the Heine-Borel property,
see proposition~\ref{prop:coerconf}.
Our definition \ref{def:conf}
is intrinsic (i.e.~coordinate-independent)
and it is practicable.
Indeed,
its direct implementation
has allowed us to 
put forward novel examples
of G1 stable ghost systems
that could not have been obtained
with the artillery in the literature,
see (\ref{eq:exth3no2}).

Moreover,
we have explicitly related confining functions
to the Liouville-Arnold theorem~\cite{Sardana}.
In lemma \ref{lem:involution},
we have proved that,
given an integrable mechanical system,
the presence of one confining conserved quantity
guarantees the fulfillment of all hypotheses
and ensures that the phase space is foliated 
exclusively by compact leaves.
This result
connects our work with open questions in
the mathematical field of integrable systems~\cite{Miranda}.
Its exploitation
conforms to yet another appealing
direction for inter-disciplinary research.

\paragraph{Outlook.} 
More ambitiously,
one may envisage a change of paradigm
in the analytic proof of global stability:
without appeal to conserved quantities.
This tantalizing possibility
has been considered in~\cite{Deffayet:2023wdg}.
Deffayet and coauthors
set forth and illustrate a conjecture
in such an unexplored arena.
We deem the search for analytic proofs of global stability
without calling for a conserved quantity
pivotal to the development of realistic ghost models.
This physically motivated search seems particularly enticing
for the community studying non-integrable systems.

The application of our results to field theory and their quantization
are obvious follow ups of noteworthy physical interest.
In both contexts,
the inequivalence between global stability definitions
\ref{def:G1bounded} and \ref{def:G1compact}
is consequential. 
This disjunction
and other functional analytical subtleties
are essential aspects to bear in mind,
even in non-ghostly basic settings such as a free massless scalar field~\cite{BarberoG:2022lji,BarberoG:2024qae}.
We defer the investigation of such compelling topics to future works.

\paragraph{Acknowledgements.}
We would like to show our high appreciation to Xavier Gr\`acia,
for his suggestion to include the relation to proper functions,
for perspicacious feedback 
and for pointing out~\cite{Palais:1957,SteenSeebach,Lee:topo}.
We are indebted to
Jose Beltr\'an Jim\'enez,
Saskia Demulder,
Manuel Lainz Valc\'azar
and Alexander Vikman
for insightful conversations
at the onset of investigations.
We have benefitted from 
Ivano Basile's constructive comments on
an earlier version of this text.
VED and JGR greatly thank
Shinji Mukohyama
for enlightening discussions
during a visit to
the Yukawa Institute for Theoretical Physics
(YITP).
Additionally,
VED and JGR are deeply grateful to
all members of the Kavli Institute
for the Physics and Mathematics
of the Universe
(Kavli IPMU)
for their extended hospitality
during the development of this work;
in particular, 
host Elisa G.~M.~Ferreira
and vibrant interlocutors
Andr\'es Nahuel Briones,
Thomas Edward Melia and Tsutomu Yanagida.
\\
The work of VED,
as well as an invited visit of JGR
to the Excellence Cluster ORIGINS, are funded by
the Deutsche Forschungsgemeinschaft
(DFG, German Research Foundation),
under Germany's Excellence Strategy
-- EXC-2094 -- 390783311. 
JGR acknowledges financial support
from Ministerio de Ciencia, Innovaci\'on y Universidades in Spain,
project D2021-125515NB-21
and from Agencia Estatal de Investigaci\'on,
project RED2022-134301-T.

\appendix

\section{Supplementary proofs}
\label{app:thenergy}

This appendix proves
theorems \ref{th:energy} and \ref{th:coercive},
as well as
the one-sided implications (\ref{eq:impl1}) and (\ref{eq:impl2}).
The reader may find useful the books in topology by Schechter~\cite{Schechter} and especially by Lee~\cite{Lee:topo}.

With the aim of proving implication (\ref{eq:impl1}),
we introduce the notion of a positively coercive function.
\begin{definition}
\label{def:rescoer}
A function $f:\mathbb{R}^n\rightarrow \mathbb{R}$
is \textbf{positively coercive} if
the value of $f(x)$ tends to plus infinity as the modulus of $x$ diverges:
\begin{align}
\label{eq:rescoer}
\displaystyle\lim_{\|x\|\rightarrow +\infty}f(x)\rightarrow +\infty.
\end{align}
\end{definition}
\noindent
Notice that a positively coercive function
is an instance of a coercive function
--- recall definition \ref{def:limcoerc}.
In particular,
coercive functions with a global minimum
are positively coercive.
Implication (\ref{eq:impl1})
follows readily.
\begin{lem}
\label{lem:defposcoercive}
\emph{\textbf{[Entails implication (\ref{eq:impl1}).]}}
If a twice-differentiable function $f:\mathbb{R}^n\rightarrow \mathbb{R}$
has a minimum and a positive definite Hessian everywhere,
then it is positively coercive.
\end{lem}

\begin{proof}
Let $x_0\in\mathbb{R}^n$ be the point where $f$ has a minimum.
The Taylor expansion up to second order of $f$ around $x_0$ satisfies
\begin{align}
\label{eq:taylor}
f(x)=f(x_0)
+\nabla f(x_0)(x-x_0)
+(x-x_0)^T\Hessian f\big(x_0+a(x-x_0)\big)(x-x_0),
\end{align}
for some $a\in (0,1)$.
Because $f$ has a positive definite Hessian everywhere, $f(x)\geq f(x_0)$.
The equality only holds for $x=x_0$.
Namely,
$x_0$ is the global minimum.

Because $\mathbb{R}^n$ has the Heine-Borel property,
any ball therein is a compact set.
The corresponding boundary is a closed set inside the ball
and thus
itself is a compact set. 
Any continuous function in a compact set
achieves both an absolute maximum and an absolute minimum in the set.
In particular,
consider the boundary of the unit ball centered on $x_0$.
Let $f_{\ast}$ denote
the absolute minimum of $f$ in $\partial B(x_0,1)$.
Notice that $f_{\ast}>f(x_0)$.

By definition,
$f$ is positively coercive if,
for every $d\in\mathbb{R}^+$,
there exists
a radius $R_d$ such that,
for any point $x\in\mathbb{R}^n$ with bigger modulus $\|x\|>R_d$,
the inequality $f(x)>d$ holds true.  
Without loss of generality,
we take $d>f(x_0)$.
In the following,
we prove that the ansatz
\begin{align}
\label{eq:ansatzRd}
R_d=\frac{d-f(x_0)}{f_{\ast}-f(x_0)}+\|x_0\|+1
\end{align} 
realizes the requirement.
Observe that $d>f(x_0)$ ensures $R_d>0$.

Let $r=\|x-x_0\|$
and consider the auxiliary point 
\begin{align}
\label{eq:auxpoint}
y=\frac{x-x_0}{r}+x_0\in\partial B(x_0,1).
\end{align}
By virtue of the mean value theorem,
\begin{align}
\label{eq:posdef1}
f(y)-f(x_0)=\nabla f(z)(y-x_0),
\qquad
z\equiv y+b(x_0-y),
\end{align}
for some $b\in(0,1)$.
The Taylor expansion up to second order of $f$ around $z$ fulfills
\begin{align}
\label{eq:posdef2}
f(x)=f(z)
+\nabla f (z)(x-z)
+(x-z)^T\Hessian f\big(z+c(x-z)\big)(x-z),
\end{align}
for some $c\in(0,1)$.
Because $f$ has a positive definite Hessian everywhere,
$f(x)\geq f(z)+\nabla f (z)(x-z)$.
Since $x-z=(y-x_0)(r-1+b)$
and making use of (\ref{eq:posdef1}),
it follows that
\begin{align}
\label{eq:ineqa}
f(x)\geq f(z)+\nabla f (z)(y-x_0)(r-1+b)
=f(z)+\big[f(y)-f(x_0)\big](r-1+b).
\end{align}
Provided that $\|x\|>R_d$ in (\ref{eq:ansatzRd})
and employing the triangle inequality,
\begin{align}
\label{eq:ineqb}
r-1+b=
\|x-x_0\|-1+b>
\|x\|-\|x_0\|-1+b>
\frac{d-f(x_0)}{f_{\ast}-f(x_0)}.
\end{align}
holds true.
Substitution of (\ref{eq:ineqb}) in (\ref{eq:ineqa}),
in light of $x_0$ being a global minimum
and the very definitions of $y$ and $f_{\ast}$,
finally yields
\begin{align}
f(x)>
f(z)+\big[f(y)-f(x_0)\big]
\frac{d-f(x_0)}{f_{\ast}-f(x_0)}
>f(x_0)
+\big[f_{\ast}-f(x_0)\big]
\frac{d-f(x_0)}{f_{\ast}-f(x_0)}=d.
\end{align}
\end{proof}

Definition \ref{def:limcoerc}
of a coercive function can be straightforwardly extended to generic metric spaces.

\begin{lem}
\label{lem:bounded}
Let $M$ and $N$ be metric spaces.
A coercive function $f:M\rightarrow N$
has bounded level sets.
Moreover,
if $M$ has the Heine-Borel property,
the level sets are relatively compact with respect to $M$. 
\end{lem}

\begin{proof}
Consider an arbitrary but fixed $c\in N$.
Since $f$ is coercive,
for $\|c\|\in\mathbb{R}$
there exists a radius $R_c\in\mathbb{R}^+$ such that,
for any point $x\in M$ with bigger modulus $\|x\|>R_c$, 
the inequality $\|f(x)\|>\|c\|$ holds true.
In particular,
at any point outside the ball $B(0,R_c)\subset M$,
$f$ takes a value different from $c$.
Consequently,
the level set $f^\neg(c)\subset B(0,R_c)$ and, therefore, it is bounded.

If $M$ has the Heine-Borel property,
then balls are compact sets.
Since balls are closed in metric spaces,
then the closure of $f^\neg(c)$ is also inside $B(0,R_c)$.
The closure of $f^\neg(c)$ is compact
as a result of it being a closed set inside a compact.
\end{proof}

\noindent
The above generalization
allows us to succinctly prove
theorems \ref{th:energy} and \ref{th:coercive},
which are hereafter restated for the convenience of the reader.
\setcounter{theorem}{0}
\begin{theorem}
Given a system $\emph{Sol}$
with $\mathcal{P}=\mathbb{R}^n$,
if the energy $E:\mathbb{R}^n\rightarrow \mathbb{R}$
is a
twice differentiable
conserved quantity
that
has a minimum
and whose Hessian is positive definite everywhere,
then the system is \emph{G1}.
\end{theorem}

\begin{proof}
Let $e$ denote the constant value of $E$    
along a fixed solution $\gamma$.
It follows that $\gamma(I_\gamma)\subset E^{\neg}(e)$.

Because $E$ is twice differentiable,
has a minimum and has a positive definite Hessian everywhere,
by virtue of lemma \ref{lem:defposcoercive},
it is coercive.
By lemma \ref{lem:bounded},
its level sets are both bounded
and relatively compact with respect to $\mathbb{R}^n$,
which has the (crucial) Heine-Borel property.
Ergo the system is G1
in the equivalent sense
of definitions \ref{def:G1bounded} and \ref{def:G1compact}.
\end{proof}

\begin{theorem}
Given a system $\emph{Sol}$
with $\mathcal{P}=\mathbb{R}^n$,
if there exists a conserved quantity
$Q:\mathbb{R}^n\rightarrow \mathbb{R}$
that is coercive,
then the system is \emph{G1}.
\end{theorem}

\begin{proof}
Let $q$ denote the constant value of $Q$    
along a fixed solution $\gamma$.
It follows that $\gamma(I_\gamma)\subset Q^{\neg}(q)$.

Because $Q$ is coercive,
by virtue of lemma \ref{lem:bounded},
its level sets are both bounded
and relatively compact with respect to $\mathbb{R}^n$,
which has the (crucial) Heine-Borel property.
Ergo the system is G1
in the equivalent sense
of definitions \ref{def:G1bounded} and \ref{def:G1compact}.
\end{proof}

Lastly,
we prove implication (\ref{eq:impl2}),
which holds true for metric spaces,
as long as 
the source space has the Heine-Borel property.

\begin{proposition}
\label{prop:coerconf}
\emph{\textbf{[Entails implication (\ref{eq:impl2}).]}}
Let $M$ be a metric space with the Heine-Borel property
and let $N$ be a metric space.
A coercive function $f:M\rightarrow N$ is confining. 
\end{proposition}

\begin{proof}
Consider an arbitrary but fixed $c\in N$.
Let $A$ denote a path-connected component of the level set $f^\neg(c)$.
The latter satisfies $A\subset f^\neg(c)$,
which is bounded as per lemma \ref{lem:bounded}.
As such, 
$f^\neg(c)\subset M$ is contained in a certain ball.
Since $M$ has the Heine-Borel property, 
balls are compact.
Moreover,
as a metric space,
$M$ is Hausdorff.
Thus,
$A$ is relatively compact with respect to $M$
because it is inside a compact.
\end{proof}

\section{Relation to proper functions}
\label{app:proper}

In topology and in differential geometry,
there exists the appurtenant concept of a {\it proper function}.
This is closely related to the notion of coercivity,
see proposition~\ref{lem:coerciveproper}.
Crucially for us,
proper functions have compact level sets.
What is more,
there are useful results in the literature
providing sufficient conditions for a function to be proper,
see e.g.~proposition 4.93 in~\cite{Lee:topo}.
Altogether,
we find proper functions to be
a promising addition to the study of G1 global stability,
in the sense of definition \ref{def:G1compact}.
In the following,
we relate our work with proper functions. 

Proper functions suffer from an occasional polysemy in the literature.
We employ the meaning in~\cite{Lee:topo}:

\begin{definition}
\label{def:proper}
Let $M$ and $N$ be topological spaces.
A function $f:M\rightarrow N$ is \textbf{proper}
if the anti-image of each compact set of $N$ is compact.
\end{definition}

\noindent
Understood in this precise sense and
as announced,
proper functions are closely related to coercive functions.

\begin{proposition}
\label{lem:coerciveproper}
Let $M$ and $N$ be metric spaces with the Heine-Borel property.
A continuous function $f:M\rightarrow N$
is coercive iff
it is proper.
\end{proposition}

\begin{proof}
\textbf{Coercive $\Rightarrow$ proper}.
Consider an arbitrary compact set $K\subset N$.
As such, it is both closed and bounded.   
Because $f$ is continuous,
the first property implies that $f^\neg(K)$ is also closed.
The second property entails the existence of
a distance  $d_K\in\mathbb{R}^+$
such that $K\subset B(0,d_K)$.
Since $f$ is coercive, 
there exists a corresponding radius $R_{d_K}$ such that, 
for all $x\in M$ with $\|x\|>R_{d_K}$, 
the inequality $\|f(x)\|>d_K$ holds true. In particular, 
\begin{align}
\label{eq:coerprop}
f^\neg(K)\subset f^{\neg}(B(0,d_K))\subset B(0,R_{d_K}).
\end{align}
The ball $B(0,R_{d_K})$ is compact because $M$ has the Heine-Borel property. Consequently,
$f^\neg(K)$ is a compact set:
it is a closed set inside a compact.
    
\textbf{Proper $\Rightarrow$ coercive.}
Given an arbitrary distance $d\in\mathbb{R}^+$,
the ball $B(0,d)\subset N$ is a compact set
because $N$ has the Heine-Borel property. 
Since $f$ is proper,
$f^\neg(B(0,d))$ is a compact set in $M$.
As such,
it is bounded.
Namely,
there exists a radius $R_d$
such that $f^\neg(B(0,d))\subset B(0,R_d)$.
Equivalently,
for all $x\in M$ with $\|x\|>R_d$,
the inequality $\|f(x)\|>d$ holds true.
\end{proof}

Proper functions are relevant to our work
owing to the following two results.
First,
following straightforwardly from definition \ref{def:proper}:
the level sets of a proper function are compact.
Second,
holding true
provided that the initial space is Hausdorff:

\begin{proposition}
\label{lem:coerconf}
Let $\mathcal{P}$ be a Hausdorff space
and let $N$ be a topological space.
A proper function $f:\mathcal{P}\rightarrow N$ is confining. 
\end{proposition}
\begin{proof}
Consider an arbitrary but fixed point $c\in N$.
Let $A$ denote a path-connected component of the level set $f^\neg(c)$.
The latter satisfies $A\subset f^\neg(c)$,
which is compact by definition of a proper function.
Since $M$ is Hausdorff
and $A$ is inside a compact,
$A$ is relatively compact with respect to $M$.
\end{proof}

\noindent
Counterexamples of the converse are (\ref{eq:exth3no2}).

\section{Integrability in light of a confining conserved quantity}
\label{app:integrability}

For completeness,
we begin by introducing the relevant geometric structures
of a mechanical system.
A more general and detailed presentation
can be found in chapter 2 of~\cite{Munoz-Lecanda:2024uyr}.
Alternatively,
see the last chapter in~\cite{Lee}.

Let $\mathcal{P}$ denote a smooth manifold,
which is Hausdorff
by definition~\cite{Lee}.
Further,
let $\omega$ be a symplectic form on $\mathcal{P}$.
Any (at least $C^2$) function $F:\mathcal{P}\rightarrow \mathbb{R}$
has an associated vector field $X_F$,
called the Hamiltonian vector field of $F$,
defined as
\begin{align}
\label{eq:infsym}
\iota_{X_F}\omega=\textrm{d}F.
\end{align}
When $F$ is a conserved quantity $F=Q$,
$X_Q$ is called an infinitesimal symmetry~\cite{Munoz-Lecanda:2024uyr}.
When $F$ is the Hamiltonian (or Lagrangian energy) of the dynamical system $F=H$,
the integral curves of $X_H$ are the solutions of the system.

The assignation of $X_F$ to $F$
induces a canonical Poisson bracket.
For two functions $F,G:\mathcal{P}\rightarrow \mathbb{R}$, 
\begin{align}
\label{eq:Poisson}
\{F,G\}:=
\iota_{X_G}\iota_{X_F}\omega=
\iota_{X_G}\text{d}F=
X_G(F).
\end{align}
When (\ref{eq:Poisson}) vanishes,
$F$ and $G$ are said to be in involution.

Recount definition \ref{def:conf}
of a confining function
and the definition of a complete vector field
--- restated below definition \ref{def:G2formal}.
The following lemma shows that
a confining function
can only be in involution with
functions whose Hamiltonian vector field
generates G1 dynamics,
in the compact sense of definition \ref{def:G1compact}.
This entails the completeness of such Hamiltonian vector fields.

\begin{lem}
\label{lem:involution}
Let $\omega$ be a symplectic form on a smooth manifold $\mathcal{P}$.
Consider two $C^2$ functions $F,G:\mathcal{P}\rightarrow \mathbb{R}$ such that $F$ is confining.
If $F$ and $G$ are in involution,
then the integral curves of $X_G$
are relatively compact with respect to $\mathcal{P}$
and,
hence,
$X_G$ is complete.
\end{lem}
\begin{proof}
If $F$ and $G$ are in involution
\begin{align}
\label{eq:invol}
0=\{F,G\}=X_G(F),
\end{align}
then $F$ is constant along the integral curves $\{\gamma\}$ of $X_G$.
In particular,
each $\gamma$ is inside a path-connected component of a level set of $F$,
which is relatively compact with respect to $\mathcal{P}$
because $F$ is confining.
Since $\mathcal{P}$ is Hausdorff,
each $\gamma$ is inside a compact set
and thus its closure $\overline{\gamma}$ is compact.
By negation of the Escape Lemma (see, e.g. \cite{Lee}),
$X_G$ is complete.
\end{proof}

\noindent
This result is relevant for integrable systems,
in the sense of section \ref{par:relationLAT}.

The celebrated Liouville-Arnold theorem
ensures the existence of action-angle coordinates for integrable systems
that trivialize the dynamics,
provided that the infinitesimal symmetries are complete vector fields.
In more detail,
the theorem ascertains that the phase space of the system is foliated by
(the path-connected components of)
the intersection of the level sets of the conserved quantities. 
Furthermore,
these leaves are diffeomorphic to $\mathbb{T}^k\times\mathbb{R}^{n-k}$,
and there exist
action-angle coordinates 
in terms of which the dynamics correspond to
that of $n$ free particles.
(For a rigorous formulation, see~\cite{Sardana}.)

Assume an integrable system 
defined on $\mathcal{P}$
where one of the conserved quantities $Q$ is confining.
By virtue of lemma \ref{lem:involution},
the infinitesimal symmetries of the conserved quantities in involution
are guaranteed to be complete.
Thus,
the hypotheses of the Liouville-Arnold theorem are satisfied.
Moreover,
since the path-connected components of the level sets of $Q$
are relatively compact with respect to $\mathcal{P}$,
the leaves provided by the Liouville-Arnold theorem are {\it compact}
(i.e.~$k=n$).
It follows that
the dynamics,
described in terms of action-angle variables,
is that of free particles moving {\it exclusively} around torii.

In brief,
theorem \ref{th:conf}
--- and its corollaries theorems \ref{th:energy} and \ref{th:coercive} ---
accept an interpretation as methods to enforce
the completeness of vector fields
associated with symmetries
and the compactness of the leaves
in integrable systems.

\section{Supplementary properties of the system (\ref{eq:exLag})}
\label{app:properties}

In this appendix,
we discuss properties of the system (\ref{eq:exLag})
against the backdrop of section \ref{sec:symmetry}.

\paragraph{Unnatural yet Lagrangian-preserving symmetry of $Q$.}
The symplectic form induced by the Lagrangian \ref{eq:exLag} is
\begin{align}
\label{eq:symforme}
\omega_L=\text{d} x\wedge\text{d}\dot{x}-\text{d} y\wedge\text{d}\dot{y},
\end{align}
and the Hamiltonian function of the system is the Lagrangian energy (\ref{eq:exE}).

The infinitesimal symmetry (\ref{eq:infsym}) associated to
the conserved quantity (\ref{eq:exQ})
via Noether's theorem
is the vector field
\begin{align}
\label{eq:exvfield}
X_Q=
(\dot{x}+c_1\dot{y})\frac{\partial}{\partial x}
-(\dot{y}+c_1\dot{x})\frac{\partial}{\partial y}
-\big[\omega^2(x+c_1y)+c_2f'\big]\frac{\partial}{\partial \dot{x}}
+\big[\omega^2(y+c_1x)+\alpha c_2f'\big]\frac{\partial}{\partial \dot{y}}.
\end{align}
The associated infinitesimal changes
for positions and velocities are
\begin{align}
\label{eq:infchange}
\begin{cases}
x\rightarrow x^\prime=
x+(\dot{x}+c_1\dot{y})\varepsilon, \\
y\rightarrow y^\prime=
y-(\dot{y}+c_1\dot{x})\varepsilon, 
\end{cases}
\qquad
\begin{cases}
\dot{x}\rightarrow \dot{x}^\prime=
\dot{x}-\big[\omega^2(x+c_1y)+c_2f'\big]\varepsilon, \\
\dot{y}\rightarrow \dot{y}^\prime=
\dot{y}+\big[\omega^2(y+c_1x)+\alpha c_2f'\big]\varepsilon,
\end{cases}   
\end{align}
respectively.
Here, $\varepsilon$ is an infinitesimal parameter.
Notice that the changes in positions depend on velocities:
the signature of an unnatural symmetry.

Being a Noether symmetry,
(\ref{eq:exvfield}) leaves invariant
the Lagrangian energy and the symplectic form.
However,
it need not be a symmetry of the Lagrangian.
In this case,
it does turn out to be one:
\begin{align}
\mathcal{L}_{X_Q}L=-2\frac{\text{d}}{\text{d} t}\left[\frac{\omega^2}{2}\left(x^2+y^2+2c_1xy\right)
+c_2f(x+\alpha y)\right].
\end{align}

\paragraph{Bi-Hamiltonian nature of the system.}
A mechanical system is (quasi-)bi-Hamiltonian
(in the sense of~\cite{Brouzet})
if there exists
a compatible symplectic form $\hat{\omega}$
and a function $\hat{H}$
with the same dynamics as the original pair $(\omega_L,E_L)$.

That is the case for the system (\ref{eq:exLag}),
whenever $\alpha\neq\pm1$\footnote{For $\alpha=\pm1$, 
the two-form $\hat{\omega}$ is not regular,
i.e.~it is presymplectic.
This case is not further developed because
the G1 stable examples in sections \ref{sec:ex1}-\ref{sec:ex3}
require $\alpha\neq\pm1$.}. 
Then,
the compatible symplectic form is
\begin{align}
\hat{\omega}=
\text{d} x\wedge \text{d} \dot{x}
+c_1\text{d} x\wedge \text{d} \dot{y}
+\text{d} y\wedge \text{d} \dot{y}
+c_1\text{d} y\wedge \text{d} \dot{x}
\end{align}
and $\hat{H}=Q$ given by (\ref{eq:exQ}).
Explicitly,
the vector fields that encode the dynamics in both settings,
\begin{align}
\iota_{X_{E_L}}\omega_L=\textrm{d}E_L
\,\, \textrm{ and } \,\,
\iota_{\hat{X}_{Q}}\hat{\omega}=\textrm{d}Q,
\end{align}
are equal: $X_{E_L}=\hat{X}_{Q}$.
Besides,
both structures are compatible because
the technical requirement that the Nijenhuis torsion vanishes is met.



\end{document}